%% file: main.tex
\documentclass[runningheads]{llncs}

\input{packages}
\newboolean{final}
\setboolean{final}{false}

\newcommand\nicoleta[1]{\textcolor{blue}{Nicoleta: #1}}

\begin{document}

\ifthenelse{\boolean{final}}{
    \title{Query Rewriting On Path Views\\Without Integrity Constraints}}
    {\title{Query Rewriting On Path Views\\Without Integrity Constraints}}
%
%
\author{Julien Romero\inst{1}\orcidID{0000-0002-7382-9077} \and
Nicoleta Preda\inst{2} \and
Fabian Suchanek\inst{1}}
\authorrunning{J. Romero et al.}
%
\institute{LTCI, Télécom Paris, Institut Polytechnique de Paris
\email{\{julien.romero,fabian.suchanek\}@telecom-paris.fr} \and
University of Versailles
\email{nicoleta.preda@uvsq.fr}}
\maketitle              
\begin{abstract}
A view with a binding pattern is a parameterised query on a database. Such views are used, e.g., to model Web services. To answer a query on such views, one has to orchestrate the views together in execution plans. The goal is usually to find equivalent rewritings, which deliver precisely the same results as the query on all databases. However, such rewritings are usually possible only in the presence of integrity constraints – and not all databases have such constraints.
In this paper, we describe a class of plans that give practical guarantees about their result even if there are no integrity constraints. We provide a characterisation of such plans and a complete and correct algorithm to enumerate them. Finally, we show that our method can find plans on real-world Web Services.

\end{abstract}

\input{parts/introduction}
\input{parts/related_work.tex}

\input{parts/preliminaries.tex}

\input{parts/smart-definition}
\input{parts/smart-recognition-naive}

\input{parts/smart-recognition-grammar-fabian}
\input{parts/algorithm.tex}
\input{appendix/smart-generation}

\input{parts/experiments}

\input{parts/conclusion}

\bibliographystyle{splncs04}
\bibliography{acmart.bib}

\ifthenelse{\boolean{final}}{}{
\clearpage

\appendix
\input{appendix/appendix.tex}}

\end{document}

%% file: packages.tex
\usepackage[utf8]{inputenc}
\usepackage[T1]{fontenc}
\usepackage{xcolor}
\usepackage{graphicx}
\usepackage{hyperref}

\graphicspath{ {images/} }

\usepackage{amsmath}
\usepackage{booktabs} 
\usepackage{tikz-cd} 
\usepackage{pgfplots}
\usetikzlibrary{patterns}
\usepackage{balance}
\usepackage{subcaption}
\pgfplotsset{compat=1.11,
        /pgfplots/ybar legend/.style={
        /pgfplots/legend image code/.code={%
        \draw[##1,/tikz/.cd,bar width=3pt,yshift=-0.2em,bar shift=0pt]
                plot coordinates {(0cm,0.8em)};},
},
}

\usepackage{amsthm} 
\usepackage{enumitem}
\usepackage{pgf}

\usetikzlibrary{calc}
\usetikzlibrary{decorations.pathreplacing}
\usetikzlibrary{shapes,arrows}
\usetikzlibrary{shadows}
\usetikzlibrary{decorations.pathmorphing} 
\usetikzlibrary{fit}					
\usetikzlibrary{backgrounds}	

\usepackage{thmtools}
\declaretheorem[name=Definition]{Definition}
\numberwithin{Definition}{section}
\declaretheorem[name=Property,sibling=Definition]{Property}
\numberwithin{Property}{section}

\numberwithin{Hypothesis}{section}
\declaretheorem[name=Theorem,sibling=Definition]{Theorem}
\numberwithin{Theorem}{section}

\numberwithin{Corollary}{section}
\declaretheorem[name=Example,sibling=Definition]{Example}
\numberwithin{Example}{section}

\numberwithin{Proposition}{section}
\declaretheorem[name=Lemma,sibling=Definition]{Lemma}
\numberwithin{Lemma}{section}

\numberwithin{Grammar}{section}

\numberwithin{Algorithm}{section}

\numberwithin{Assumption}{section}

\let\oldnl\nl
\newcommand{\noln}{\renewcommand{\nl}{\let\nl\oldnl}}

\newcommand{\ignore}[1]{}

\renewcommand{\paragraph}[1]{\noindent\textbf{#1.}}

\numberwithin{equation}{section}

\usepackage{listings}
\usepackage[linesnumbered,ruled,vlined]{algorithm2e}


%% file: parts/introduction.tex
\section{Introduction}

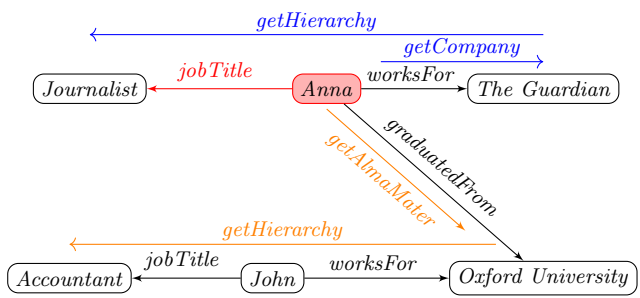
\begin{figure}
    \centering
    \input{images/CIKM2018/fig-album3}
    \caption{An equivalent execution plan (blue) and a maximal contained rewriting (orange) executed on a database instance (black). \label{fig:ex1}}
\end{figure}

A view with binding patterns is a parameterised query defined in terms of a global schema~\cite{HalevyAnswering2001}.
Such a query works like a function: it requires specific values as input and delivers the query results as output. For example, consider the database instance about employees at  Figure~\ref{fig:ex1}. The call to the function \emph{getCompany} with an employee \textit{Anna} as input, returns the company \textit{The Guardian} as output. Abstractly, the function is represented as the rule: $\textit{getCompany}(\textit{in}, \textit{out}) \leftarrow \textit{worksFor}(\textit{in}, \textit{out})$. The \textit{worksFor} relation is of the global schema, which is orthogonal to the schema of the actual data.  
Unlike query interfaces like SPARQL endpoints, functions prevent arbitrary access to the database engines. In particular, one can model Web forms or REST Web Services as views with binding patterns. 
According to \href{http://programmableweb.com}{programmableweb.com}, there are currently more than 22.000 such REST Web Services.

If we want to answer a query on a global database that can be accessed only through functions, we have to orchestrate the functions into an execution plan.
In our example from Figure~\ref{fig:ex1}, if we want to find the job title of Anna, we first have to find her company (by calling \emph{getCompany}), and then her job title (by calling \emph{get\-Hierarchy} on her company, and filtering the results about Anna). Our problem is thus as follows: Given a user query (such as \emph{job\-Title(Anna, $x$))} and a set of functions (each being a parameterised conjunctive query), find an execution plan (i.e., a sequence of function calls) that delivers the answer to the query on a database that offers these functions. 
While the schema of the database is known to the user, she or he does not know whether the database contains the answer to the query at all.

Much of the literature concentrates on finding \emph{equivalent rewritings}, i.e., execution plans that deliver the same result as the original query on all databases that offer this specific set of functions. Unfortunately, our example plan is not an equivalent rewriting: it will deliver no results on databases where (for whatever reasons) Anna has a job title but no employer. The plan is equivalent to the query only if an integrity constraint stipulates that every person with a job title must have an employer and the database instance is complete.

Such constraints are hard to come by in real life, because they may not hold (a person can have a job title but no employer; a person may have a birth date but no death date; some countries do not have a capital\footnote{e.g., the Republic of Nauru}; etc.). Even if they hold in real life, they may not hold in the database due to the incompleteness of the data. Hence, they are also challenging to mine automatically. In the absence of constraints, however, an atomic query has an equivalent rewriting only if there is a function that was defined precisely for that query.

 \emph{Maximally contained rewritings} are preferred to equivalent rewritings in data integration systems.  There, the databases are incomplete, and the equivalent rewritings usually fail to deliver results. Intuitively speaking, maximally contained rewritings are execution plans that try to find all calls that could potentially lead to an answer. In our example, the plan \emph{get\-Alma\-Mater}, \emph{get\-Hierarchy} is included in the maximally contained rewriting: It asks for the university where Anna graduated, and for their job positions. If Anna happens to work at the university where she graduated, this plan will answer the query.

This plan appears somehow less reasonable that our first plan because it works only for people who work at their alma mater. However, both plans are equal concerning their formal guarantees: none of them can guarantee to deliver the answers to the query. This is a conundrum: Unless we have information about the data distribution or more schema information, we have no formal means to give the first plan higher chances of success than the second plan -- although the first plan is intuitively much better.

In this paper, we propose a solution to this conundrum: 
We can show that the first plan (\emph{get\-Company}, \emph{get\-Hierarchy}) is ``smart'', in a sense that we formally define. We can give guarantees about the results of smart plans in the absence of integrity constraints. We also give an algorithm that can enumerate all smart plans for a given \textbf{atomic query} and \textbf{path-shaped functions} (as in Figure~\ref{fig:ex1}). We show that under a condition that we call the \emph{Optional Edge Semantics} our algorithm is complete and correct, i.e., it will exhaustively enumerate all such smart plans. We apply our method to real Web services and show that smart plans work in practice and deliver more query results than competing approaches.

This paper is structured as follows: Section~\ref{sec:cikm-rel} discusses related work, Section~\ref{sec:cikm-prel} introduces preliminaries, and Section~\ref{sec:cikm-def} gives a definition of smart plans. Section~\ref{sec:cikm-rec} provides a method to characterise smart plans, and Section~\ref{sec:generation-main} shows there exists an algorithm that can generate smart plans. We provide extensive experiments on synthetic and real Web services to show the viability of our method in Section~\ref{sec:cikm-experiments}.

\ifthenelse{\boolean{final}}{
    All the proofs and technical details are in the appendix of the accompanying technical report~\cite{technicalreport}.   
}
{
All the proofs and technical details are in the appendix.
}

%% file: images/CIKM2018/fig-album3.tex
\tikzstyle arrowstyle=[blue,semitransparent,scale=2]
\tikzstyle basiclabel=[draw=none,fill=none,shape=rectangle,inner sep=2pt,scale=.8]
\tikzstyle leftlabel=[basiclabel,anchor=east]
\tikzstyle rightlabel=[basiclabel,anchor=west]
\tikzstyle bottomlabel=[basiclabel,anchor=north]
\tikzstyle toplabel=[basiclabel,anchor=south]

\definecolor{lightblue}{cmyk}{0.12,0,0,0.15} 
\tikzstyle{block} = [rectangle, draw, fill=white, 
text centered, rounded corners]

\definecolor{darkpastelgreen}{rgb}{0.01, 0.75, 0.24}
\definecolor{dartmouthgreen}{rgb}{0.05, 0.5, 0.06}
\definecolor{forestgreen}{rgb}{0.0, 0.27, 0.13}
\definecolor{lasallegreen}{rgb}{0.03, 0.47, 0.19}

\tikzstyle{headvar} = [rectangle, draw, fill=lightblue!70,
text centered, rounded corners]    
\tikzstyle{labelcase} = [rectangle,    text centered]
\tikzstyle{background}=[rectangle,   fill=lightblue!70,
inner sep=0.2cm,
rounded corners=5mm]

\tikzstyle{line} = [draw, -latex']
\tikzstyle{linenoarrow} = [draw]
\tikzstyle{invisibleline} = [-latex',sloped]
\tikzstyle{dashedline} = [draw, dashed]
\tikzstyle{inputtar} = [rectangle, draw, fill=lightblue!70, 
text centered, rounded corners]

\tikzstyle{inputvar} = [rectangle, draw=red,  fill=red!30,
text centered, rounded corners]

\begin{tikzpicture}[scale=0.9, transform shape]

\node [inputvar] (z) {\textit{Anna}};
\node [block] (y) at ($(z)+(+3.2cm,0cm)$) {\textit{The Guardian}};
\node [block] (x) at ($(z)+(-3.5cm,0cm)$) {\textit{Journalist}};
\node [block] (t) at ($(y)+(-0.0cm,-2.8cm)$){\textit{Oxford University}};
\path [line] (z) --node [above,align=center] {{\textit{worksFor}} } (y);
\path [line,red] (z) --node [above,align=center] {{\textit{jobTitle}} } (x);
\path [line] (z) --node [sloped,above, align=center] {{\textit{graduatedFrom}} } (t);

\coordinate (zu) at   ($(z) + (+0.8cm,+0.4cm) $); 
\coordinate (yu) at   ($(y) + (+0cm,+0.4cm) $); 
\draw [blue,->]     (zu) to  (yu);
\node [labelcase,blue] (lf1) at   ($0.5*(zu) + 0.5*(yu) + (0,+0.2cm)$) {\textit{getCompany}};

\coordinate (xb) at   ($(x) + (+0cm,+0.8cm) $);
\coordinate (yb) at   ($(y) +  (+0cm,+0.8cm) $);
\draw [->,blue]    (yb) to (xb);
\node [labelcase,blue] (lf2) at   ($0.5*(xb) + 0.5*(yb) + (0cm,+0.19cm)$) {\textit{getHierarchy}};

\node [block] (z2) at ($(t)+(-4cm,0cm)$) {\textit{John}};
\node [block] (x2) at ($(z2)+(-3cm,0cm)$) {\textit{Accountant}};

\path [line] (z2) --node [above,align=center] {{\textit{worksFor}} } (t);
\path [line] (z2) --node [above,align=center] {{\textit{jobTitle}} } (x2);

\coordinate (xb2) at   ($(x2) + (+0cm,+0.5cm) $);
\coordinate (yb2) at   ($(t) +  (-0.7cm,+0.5cm) $);
\draw [->, orange]    (yb2) to (xb2);
\node [labelcase,orange] (lf22) at   ($0.5*(xb2) + 0.5*(yb2) + (0cm,+0.2cm)$) {\textit{getHierarchy}};

\coordinate (zr) at   ($(z) + (-0cm,-0.3cm) $); 
\coordinate (tr) at   ($(yb2) + (-0.45cm,0.2cm) $); 
\path [line, orange] (zr) --node [sloped,below, align=center] {{\textit{getAlmaMater}} } (tr);


\end{tikzpicture}

%% file: parts/related_work.tex
\section{Related Work}\label{sec:cikm-rel}

\paragraph{Equivalent Rewritings} An equivalent rewriting of a query is an alternative formulation of the query that has the same results as the query on all databases.
Equivalent rewritings have also been studied in the context of views with binding patterns~\cite{benedikt2016generating,romero2020equivalent}. However,
they may not be sufficient to answer the query \cite{HalevyAnswering2001}. Equivalent rewritings rely on integrity constraints, which may not be available. These constraints are difficult to mine, as most real-life rules have exceptions. Also, equivalent rewritings may falsely return empty answers only because the database instance is incomplete with respect to the integrity constraints. We aim to come up with new relevant rewritings that still offer formal guarantees about their results. 

\paragraph{Maximally Contained Rewriting} In data integration applications, where databases are incomplete, and equivalent rewritings are likely to fail, maximally contained rewritings have been proposed as an alternative. 
A maximally contained rewriting is a query expressed in a chosen language that retrieves the broadest possible set of answers~\cite{HalevyAnswering2001}.
By definition, the task does not distinguish between intuitively more reasonable rewritings and rewritings that stand little chance to return a result on real databases. For views with binding patterns, the problem has been studied for different rewriting languages and under different constraints~\cite{DuschkaL97,CaliICDE2008,NashL04}. 
Some works~\cite{angie,susie} propose to prioritise the execution of the calls in order to produce the first results fast. While the first work~\cite{angie} does not give guarantees about the plan results, the second one~\cite{susie} can give guarantees only for very few plans. Our work is much more general and includes all the plans generated by~\cite{susie}, as we will see.


\paragraph{Plan Execution} Several works study how to optimise given execution plans~\cite{Srivastava06,WEB_COMP_DATA_INTEGRATION_VLDB2005}. Our work, in contrast, aims at \emph{finding} such execution plans. 

\paragraph{Federated Databases} In federated databases~\cite{BuronGMM20,AebeloeMH19},
a data source supports any queries in a predefined language. In our setting, in contrast, the database can be queried only through \emph{functions}, i.e., specific predefined queries with input parameters.



%% file: parts/preliminaries.tex
\section{Preliminaries} \label{sec:cikm-prel}

We use the terminology of~\cite{romero2020equivalent}, and recall the definitions briefly.

\paragraph{Global Schema} We assume a set $\mathcal{C}$ of constants and a set $\mathcal{R}$ of binary relation names.
A \textit{fact} $r(a, b)$ is formed from a relation name $r \in \mathcal{R}$ and two constants $a, b \in \mathcal{C}$. A
\textit{database instance} $I$, or simply \textit{instance}, is a set of facts. 

\paragraph{Queries} An \textit{atom} takes the form $r(\alpha,\beta)$, where $r \in \mathcal{R}$, and $\alpha$ and $\beta$ are either constants or variables. It can be equivalently written as $r^-(\beta,\alpha)$. 
A \emph{query} takes the form:
\[q(\alpha_1,...,\alpha_m) \leftarrow  B_1,...,B_n\]
where $\alpha_1,...\alpha_m$ are variables, each of which must appear in at least one of the body atoms $B_1,...B_n$. We assume that queries are \emph{connected}, i.e., each body atom must be transitively linked to every other body atom by shared variables.

An \emph{embedding} for a query $q$ on a database instance $I$ is a substitution $\sigma$ for the variables of the body atoms so that $\forall B\in\{B_1,...,B_n\}: \sigma(B)\in I$. A \emph{result} of a query is an embedding projected to the variables of the head atom. We write $q(\alpha_1,...,\alpha_m)(I)$ for the results of the query on $I$. An \emph{atomic query} is a query that takes the form $q(x) \leftarrow r(a, x)$, where $a$ is a constant and $x$ is a variable.
A \textit{path  query} is a query of the form:
\[q(x_i)\leftarrow r_1(a,x_1), r_2(x_1,x_2), ..., r_{n}(x_{n-1},x_n)\]
where $a$ is a constant, $x_i$ is the output variable,
each $x_j$ except $x_i$ is either a variable or the constant~$a$, and $1 \leq i \leq n$.

\paragraph{Functions}
\label{sec:function-definition}
We model functions as views with binding patterns~\cite{BINDING-PATTERNS-PODS1995}, namely:
\[f(\underline{x}, y_1,...,y_m) \leftarrow  B_1,...,B_n\]
Here, $f$ is the function name, $x$ is the \emph{input variable} (which we underline), $y_1,...,y_m$ are the \emph{output variables}, and any other variables of the body atoms are \emph{existential variables}. 
In this paper, we are concerned with \emph{path functions}, where the body atoms are ordered in a sequence $r_1(\underline{x}, x_1), r_2(x_1,x_2),...,r_n(x_{n-1}, x_n)$. The first variable of the first atom is the input of the plan, the second variable of each atom is the first variable of its successor, and the output variables follow the order of the atoms.

\noindent \emph{Calling} a function for a given value of the input variable means finding the result of the query given by the body of the function on a database instance.

\paragraph{Plans}
A \emph{plan} takes the form 
\[\pi(x) = c_1, \ldots, c_n, \gamma_1=\delta_1, \ldots, \gamma_m=\delta_m\] 
Here, $a$ is a constant and $x$ is the output variable. Each $c_i$ is a \emph{function call} of the form $f(\underline{\alpha}, \beta_1, \ldots, \beta_n)$, where $f$ is a function name, the input $\alpha$ is either a constant or a variable occurring in some call in $c_1, \ldots, c_{i-1}$, and the outputs $\beta_1, \ldots, \beta_n$ are variables. Each $\gamma_j=\delta_j$ is called a \emph{filter}, where $\gamma_j$ is an output variable of any call, and $\delta_j$ is either a variable that appears in some call or a constant.
If the plan has no filters, then we call it \emph{unfiltered}.
The \textit{semantics} of the plan is the query:
\[
q(x) \leftarrow \phi(c_1), \ldots, \phi(c_n), \gamma_1=\delta_1, \ldots, \gamma_m=\delta_m
\]
where each $\phi(c_i)$ is the body of the query defining the function $f$ of the call~$c_i$ in which we have substituted the constants and variables used in $c_i$. We have used fresh existential variables across the different $\phi(c_i)$, where $x$ is the output variable of the plan, and where $\cdot =\cdot$ is an atom that holds in any database instance if and only if its two arguments are identical.

To \emph{evaluate} a plan on an instance means running the query above. In practice, this boils down to calling the functions in the order given by the plan.  Given an execution plan $\pi_a$ and a database $I$, we call $\pi_a(I)$ the answers of the plan on $I$. 


\begin{Example}\label{ex1} Consider our example in Figure~\ref{fig:ex1}. There are 3 relation names in the database: \textit{worksFor}, \textit{jobTitle}, and \textit{graduatedFrom}. The functions are:
\begin{align*}
\textit{getCompany}(\underline{x},y)  & \leftarrow \textit{worksFor}(\underline{x}, y) \\
\textit{getHierarchy}(\underline{y},x,z)  & \leftarrow \textit{worksFor}^-(\underline{y},x), \textit{jobTitle}(x,z) \\
\textit{getEducation}(\underline{x}, y)  & \leftarrow \textit{graduatedFrom}(\underline{x},y)
\end{align*}

The following is an execution plan:
\begin{align*}
\pi_{1}(z) = & \textit{getCompany}(\underline{Anna},x), \textit{getHierarchy}(\underline{x},y,z), y=\textit{Anna}
\end{align*}
The first element is a function call to \textit{getCompany} with the name of the person (\textit{Anna}) as input, and the variable $x$ as output. 
The variable $x$ then serves as input in the second function call to \textit{get\-Hierarchy}. 
Figure~\ref{fig:ex1} shows the plan with an example instance. 
This plan computes the query:
\begin{align*}
    \textit{worksFor}(\textit{Anna},x), \textit{worksFor}^-(x, y),
    \textit{jobTitle}(y,z), y=\textit{Anna}
\end{align*}

\noindent In our example instance, we have the  embedding:
\[\sigma=\{x \xrightarrow{} \textit{The Guardian},  y \xrightarrow{} \textit{Anna}, z \xrightarrow{} \textit{Journalist} \}.
\]

\end{Example}

\noindent An execution plan $\pi$ is \emph{redundant} if it has no call using the constant $a$ as input, or if it contains a call where none of the outputs is an output of the plan or an input to another call. 

\label{def:eqv}
An \textit{equivalent rewriting} of an atomic query $q(x) \leftarrow r(a, x)$ is an execution plan that has the same results as $q$ on all database instances. For our query language, a \textit{maximally contained rewriting} for the query $q$ is a plan whose semantics contains the atom $r(a, x)$.

%% file: parts/smart-definition.tex
\section{Defining Smart Plans}\label{sec:cikm-def}

Given an atomic query, and given a set of path functions, we want to find a reasonable execution plan that answers the query.

\paragraph{Introductory Observations} Let us consider again the query $q(x) \leftarrow \textit{jobTitle}(\textit{Anna}, x)$ and the two plans in Figure~\ref{fig:ex1}. The first plan seems to be smarter than the second one. The intuition becomes more formal if we look at the queries in their respective semantics. The first plan is the plan $\pi_1(z)$ given in Example \ref{ex1}. Its semantics is the query: $\textit{worksFor}(\textit{Anna},x),\textit{worksFor}^-(x, y), \textit{jobTitle}(y, z), y=\textit{Anna}$. If the first atom has a match in the database instance, then $y=\textit{Anna}$ is indeed a match, and the plan delivers the answers of the query. If the first atom has no match in the database instance, then the plan returns no result, while the query may have one. To make the plan equivalent to the query on all database instances, we would need the following 
unary inclusion dependency: $jobTitle(x,y) \rightarrow \exists z: \textit{worksAt}(x,z)$.
In our setting, however, we cannot assume such an integrity constraint.
Let us now consider the second plan: 
\begin{align*}
\pi_{2}(z) = \textit{getEducation}(\underline{\textit{Anna}}, x), \textit{getHierarchy}(\underline{x},y,z),  y=\textit{Anna} \label{eq:oxford}
\end{align*}
\noindent Its  semantics are:  $\textit{graduatedFrom}(\textit{Anna},x),\textit{worksFor}^-(x, y), \textit{jobTitle}(y, z), y=\textit{Anna}$. To guarantee that $y=\textit{Anna}$ is a match, we need one constraint at the schema level: the inclusion dependency $\textit{graduatedFrom}(x,y) \rightarrow \textit{worksFor}(x,y)$. However, this constraint does not hold in the real world, and it is stronger than a unary inclusion dependency (which has an existential variable in the tail). Besides, $\pi_2$, similarly to $\pi_1$, needs the unary inclusion dependency $jobTitle(x,y) \rightarrow \exists z: \textit{graduatedFrom}(x,z)$ to be an equivalent rewriting.

\paragraph{Definition} In summary, the first plan, $\pi_{1}$, returns the query answers if all the calls return results. The second plan, $\pi_{2}$, may return query answers, but in most of the cases even if the calls are successful, their results are filtered out by the filter $y=\textit{Anna}$.  This brings us to the following definition of smart plans: 
\begin{Definition}[Smart Plan]\label{def:zsmart-plan}
Given an atomic query $q$ and a set of functions, a plan $\pi$ is \emph{smart} if the following holds on all database instances $I$: If the filter-free version of $\pi$ has a result on $I$, then $\pi$ delivers exactly the answers to the query.
\end{Definition}

\noindent 
We also introduce weakly smart plans:


\begin{Definition}[Weakly Smart Plan]\label{def:strong-smart-plan}
Given an atomic query $q$ and a set of functions, a plan $\pi$ is \emph{weakly smart} if the following holds on all database instances $I$ where $q$ has at least one result: If the filter-free version of $\pi$ has a result on $I$, then $\pi$ delivers a super-set of the answers to the query.
\end{Definition}

\noindent Weakly smart plans deliver a superset of the answers of the query, and thus do not actually help in query evaluation.  Nevertheless, weakly smart plans can be useful: 
For example, if a data provider wants to hide private information, like the phone number of a given person, they do not want to leak it in any way, not even among other results. Thus, they will want to design their functions in such a way that no \textit{weakly smart plan} exists for this specific query.

Every smart plan is also a weakly smart plan. Some queries will admit only weakly smart plans and no smart plans, mainly because the variable that one has to filter on is not an output variable.

\ignore{

\begin{figure*}[t]
    \centering
    \input{images/CIKM2018/fig-plan-album2}
    \caption{A smart plan for the query \textit{job\-Title}(\textit{Anna}, $?x$), which Susie will not find.\label{fig:cikm-plan-ex1}}
\end{figure*}

\paragraph{Comparison with Susie} The plans generated in Susie~\cite{susie} are smart according to our definition. However, they only represent a subset of the smart plans. 
For instance, consider again the query \textit{holdsPosition(Anna, $x$)}, and assume that, in addition to the function \emph{getHierarchy}, we have the following two functions:
\begin{align*}
   \textit{getProfessionalAddress}(\underline{x}, y, z) & \leftarrow \textit{worksFor}(\underline{x},y), \textit{locatedIn}(y,z) \\
   \textit{getEntityAtAddress}(\underline{x},y) & \leftarrow \textit{locatedIn}(\underline{x},y)
\end{align*}
\noindent Then the following plan is smart (see Figure~\ref{fig:cikm-plan-ex1}):
\begin{align*}
    \pi_{Anna}(x) = &\textit{getProfessionalAddress}(Anna, y, z),\\ &\textit{getEntityAtAddress}(\underline{z},y'), \textit{getHierarchy}(\underline{y'},t,x) 
\end{align*}
\noindent However, it will not be discovered by the algorithm in~\cite{susie}.}

\begin{figure*}[ht]
    \centering
    \input{images/CIKM2018/dangie-two-plans.tex}
    \caption{A non-smart execution plan for the query \emph{phone(Anna,$x$)}. Left: a database where the plan answers the query. Right: a database where the unfiltered plan has results, but the filtered plan does not answer the query.\label{fig:phone}}
\end{figure*}
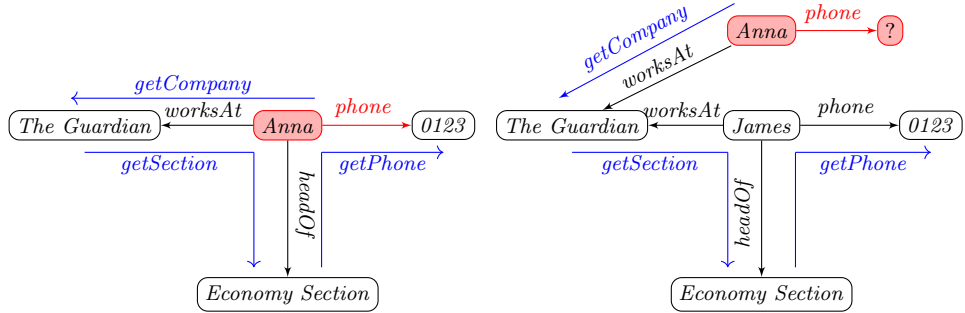

\paragraph{Smart plans versus equivalent plans} Consider again the plans $\pi_1$ (smart) and $\pi_2$ (not-smart) above. Both plans assume the existence of a unary inclusion dependency. The difference is that in addition,  $\pi_2$ relies on an additional role inclusion constraint. Is it thus sufficient to assume unary inclusion dependencies between all pairs of relations, 
and apply the algorithm in~\cite{romero2020equivalent} to find equivalent rewritings? 
The answer is no:
Figure~\ref{fig:phone} shows a plan that is equivalent if the necessary unary inclusion dependencies hold. However, the plan is not smart. On the database instance shown on the right-hand side, the unfiltered plan returns a non-empty set of results that does not answer the query. 


\paragraph{Problem} After having motivated and defined our notion of smart plans, we are now ready to state our problem: Given an atomic query, and a set of path functions, we want to enumerate all smart plans.

%% file: images/CIKM2018/fig-plan-album2.tex
\tikzstyle arrowstyle=[blue,semitransparent,scale=2]
\tikzstyle basiclabel=[draw=none,fill=none,shape=rectangle,inner sep=2pt,scale=.8]
\tikzstyle leftlabel=[basiclabel,anchor=east]
\tikzstyle rightlabel=[basiclabel,anchor=west]
\tikzstyle bottomlabel=[basiclabel,anchor=north]
\tikzstyle toplabel=[basiclabel,anchor=south]

\definecolor{lightblue}{cmyk}{0.12,0,0,0.15} 
\tikzstyle{block} = [rectangle, draw, fill=white, 
     text centered, rounded corners]
     
\tikzstyle{headvar} = [rectangle, draw, fill=lightblue!70,
     text centered, rounded corners]    
\tikzstyle{labelcase} = [rectangle,    text centered]
\tikzstyle{background}=[rectangle,   fill=lightblue!70,
                                                inner sep=0.2cm,
                                                rounded corners=5mm]

\tikzstyle{line} = [draw, -latex']
\tikzstyle{linenoarrow} = [draw]
\tikzstyle{invisibleline} = [-latex',sloped]
\tikzstyle{dashedline} = [draw, dashed]
\tikzstyle{inputtar} = [rectangle, draw, fill=lightblue!70, 
     text centered, rounded corners]

\tikzstyle{inputvar} = [rectangle, draw=red,  fill=red!30,
text centered, rounded corners]

\begin{tikzpicture}[scale=0.9, transform shape]






\node [inputvar, node distance=2.5cm] (z2) {\textit{Anna}};
\node [block] (y2) at ($(z2)+(+3.5cm,0cm)$) {\textit{The Guardian,}};
\node [block] (t2) at ($(y2)+(+4.5cm,0cm)$) {\textit{Kings Place, London}};
\node [block] (x2) at ($(z2)+(-3cm,0cm)$) {\textit{Journalist}};
\path [line] (z2) --node [above,align=center] {{\textit{worksFor}} } (y2);
\path [line,red] (z2) --node [above,align=center] {{\textit{jobTitle}} } (x2);
\path [line] (y2) --node [above,align=center] {{\textit{locatedIn}} } (t2);

\coordinate (zu2) at   ($(z2) + (-0.5cm,+0.4cm) $); 
\coordinate (tu2) at   ($(t2) + (+0cm,+0.4cm) $); 
\draw [blue,->]     (zu2) to  (tu2);
\node [labelcase,blue] (lf1) at   ($0.5*(zu2) + 0.5*(tu2) + (0,0.3cm)$) {\textit{getProfessionalAddress}};

\coordinate (xb2) at   ($(x2) + (+0cm,-0.4cm) $);
\coordinate (yb2) at   ($(y2) +  (-0.1cm,-0.4cm) $);
\draw [->,blue]    (yb2) to (xb2);
\node [labelcase,blue] (lf2) at   ($0.5*(xb2) + 0.5*(yb2) + (0cm,-0.3cm)$) {\textit{getHierarchy}};

\coordinate (tb2) at   ($(t2) + (+0cm,-0.4cm) $); 
\coordinate (yb2) at   ($(y2) + (+0.1cm,-0.4cm) $); 
\draw [blue,->]     (tb2) to  (yb2);
\node [labelcase,blue] (lf1) at   ($0.5*(tb2) + 0.5*(yb2) + (0,-0.3cm)$) {\textit{getEntityAtAddress}};

\end{tikzpicture}

%% file: images/CIKM2018/dangie-two-plans.tex
\tikzstyle arrowstyle=[blue,semitransparent,scale=2]
\tikzstyle basiclabel=[draw=none,fill=none,shape=rectangle,inner sep=2pt,scale=.8]
\tikzstyle leftlabel=[basiclabel,anchor=east]
\tikzstyle rightlabel=[basiclabel,anchor=west]
\tikzstyle bottomlabel=[basiclabel,anchor=north]
\tikzstyle toplabel=[basiclabel,anchor=south]

\definecolor{lightblue}{cmyk}{0.12,0,0,0.15} 
\tikzstyle{block} = [rectangle, draw, fill=white, 
text centered, rounded corners]

\definecolor{darkpastelgreen}{rgb}{0.01, 0.75, 0.24}
\definecolor{dartmouthgreen}{rgb}{0.05, 0.5, 0.06}
\definecolor{forestgreen}{rgb}{0.0, 0.27, 0.13}
\definecolor{lasallegreen}{rgb}{0.03, 0.47, 0.19}

\tikzstyle{headvar} = [rectangle, draw, fill=lightblue!70,
text centered, rounded corners]    
\tikzstyle{labelcase} = [rectangle,    text centered]
\tikzstyle{background}=[rectangle,   fill=lightblue!70,
inner sep=0.2cm,
rounded corners=5mm]

\tikzstyle{line} = [draw, -latex']
\tikzstyle{linenoarrow} = [draw]
\tikzstyle{invisibleline} = [-latex',sloped]
\tikzstyle{dashedline} = [draw, dashed]
\tikzstyle{inputtar} = [rectangle, draw, fill=lightblue!70, 
text centered, rounded corners]

\tikzstyle{inputvar} = [rectangle, draw=red,  fill=red!30,
text centered, rounded corners]

\begin{tikzpicture}[scale=0.9, transform shape]

\node [inputvar] (f1p1) {\textit{Anna}};
\node [block] (f1c1) at ($(f1p1)+(-3cm,0cm)$) {\textit{The Guardian}};
\node [block] (f1tel1) at ($(f1p1)+(+2.3cm,0cm)$) {\textit{0123}};

\path [line] (f1p1) --node [above,align=center] {{\textit{worksAt}} } (f1c1);
\path [line,red] (f1p1) --node [above,align=center] {{\textit{phone}} } (f1tel1);

\node [block] (f1dep1) at ($(f1p1)+(0,-2.5cm)$){\textit{Economy Section}};
\path [line] (f1p1) --node [sloped,above, align=center] {{\textit{headOf}} } (f1dep1);

\coordinate (f1xb) at   ($(f1p1) + (-0.5cm,-0.4cm) $);
\coordinate (f1yb) at   ($(f1c1) +  (+0cm,-0.4cm) $);
\coordinate (f1zb) at   ($(f1dep1) +  (-0.5cm,+0.4cm) $);
\draw [blue]    (f1yb) to (f1xb);
\draw [->,blue]    (f1xb) to (f1zb);
\node [labelcase,blue] (f1lf2) at   ($0.5*(f1xb) + 0.5*(f1yb) + (0cm,-0.18cm)$) {\textit{getSection}};

\coordinate (f1xc) at   ($(f1p1) + (+0.5cm,-0.4cm) $);
\coordinate (f1yc) at   ($(f1tel1) +  (+0cm,-0.4cm) $);
\coordinate (f1zc) at   ($(f1dep1) +  (+0.5cm,+0.4cm) $);
\draw [->,blue]    (f1xc) to (f1yc);
\draw [blue]    (f1xc) to (f1zc);
\node [labelcase,blue] (f1lf3) at   ($0.5*(f1xc) + 0.5*(f1yc) + (0cm,-0.18cm)$) {\textit{getPhone}};

\coordinate (f1p1up) at   ($(f1p1) + (+0.4cm,+0.4cm) $); 
\coordinate (f1c1up) at   ($(f1c1) + (-0.2cm,+0.4cm) $); 
\draw [blue,->]     (f1p1up) to  (f1c1up);
\node [labelcase,blue] (f1lf1) at   ($0.5*(f1p1up) + 0.5*(f1c1up) + (0,+0.2cm)$) {\textit{getCompany}};

\node [block] (p1) at ($(f1p1)+(+7cm,0cm)$) {\textit{James}};
\node [block] (c1) at ($(p1)+(-2.8cm,0cm)$) {\textit{The Guardian}};
\node [block] (tel1) at ($(p1)+(+2.5cm,0cm)$) {\textit{0123}};

\path [line] (p1) --node [above,align=center] {{\textit{worksAt}} } (c1);
\path [line] (p1) --node [above,align=center] {{\textit{phone}} } (tel1);

\node [block] (dep1) at ($(p1)+(0,-2.5cm)$){\textit{Economy Section}};
\path [line] (p1) --node [sloped,above, align=center] {{\textit{headOf}} } (dep1);

\coordinate (xb) at   ($(p1) + (-0.5cm,-0.4cm) $);
\coordinate (yb) at   ($(c1) +  (+0cm,-0.4cm) $);
\coordinate (zb) at   ($(dep1) +  (-0.5cm,+0.4cm) $);
\draw [blue]    (yb) to (xb);
\draw [->,blue]    (xb) to (zb);
\node [labelcase,blue] (lf2) at   ($0.5*(xb) + 0.5*(yb) + (0cm,-0.18cm)$) {\textit{getSection}};

\coordinate (xc) at   ($(p1) + (+0.5cm,-0.4cm) $);
\coordinate (yc) at   ($(tel1) +  (+0cm,-0.4cm) $);
\coordinate (zc) at   ($(dep1) +  (+0.5cm,+0.4cm) $);
\draw [->,blue]    (xc) to (yc);
\draw [blue]    (xc) to (zc);
\node [labelcase,blue] (lf3) at   ($0.5*(xc) + 0.5*(yc) + (0cm,-0.18cm)$) {\textit{getPhone}};

\node [inputvar] (p2) at ($(p1)+(0cm,1.4cm)$) {\textit{Anna}};
\path [line] (p2) --node [sloped,above, align=center] {{\textit{worksAt}} } (c1);

\coordinate (p2up) at   ($(p2) + (-0.4cm,+0.4cm) $); 
\coordinate (c1up) at   ($(c1) + (-0.2cm,+0.4cm) $); 
\path [line,blue] (p2up) --node [sloped,above, align=center] {{\textit{getCompany}} } (c1up);

\node [inputvar] (tel2) at ($(p2)+(+1.9cm,0cm)$) {?};
\path [line,red] (p2) --node [above,align=center] {{\textit{phone}} } (tel2);

\end{tikzpicture}

%% file: parts/smart-recognition-naive.tex
\section{Characterizing Smart Plans}\label{sec:cikm-rec}

\subsection{Web Service Functions}\label{sec:cikm-linear}

We now turn to generating smart plans. As previously stated, our approach can find smart plans only under a certain condition. This condition has to do with the way Web services work. Assume that for a given person, a function returns the employer and the address of the working place:
\[
\textit{getCompanyInfo}(\underline{x},y,z) \leftarrow \textit{worksAt}(\underline{x},y), \textit{locatedIn}(y,z)
\]

\noindent Now assume that, for some person, the address of the employer is not in the database. In that case, the call will not fail. Rather, it will return only the employer $y$, and return a null-value for the address $z$.  It is as if the atom \emph{locatedIn($y$, $z$)} were optional. 
To model this phenomenon, we introduce the notion of \emph{sub-functions}:
Given a path function $f: r_1(\underline{x_0},x_1),r_2(x_1,x_2), \ldots r_n(x_{n-1},x_n)$, the sub-function for an output variable $x_i$ is the function  $f_i(\underline{x_0}, ..., x_i) \leftarrow r_1(\underline{x_0},x_1), \ldots r_i(x_{i-1},x_i)$.

\begin{Example} The sub-functions of the function \textit{get\-Company\-Info} are $f_1(\underline{x},y) \leftarrow \textit{worksAt}(\underline{x},y)$, which is associated to $y$, and  $f_2(\underline{x},y,z) \leftarrow \textit{worksAt}(\underline{x},y), \textit{locatedIn}(y,z)$, which is associated to $z$.
\end{Example}

\noindent We can now express the Optional Edge Semantics:

\begin{Definition}[Optional Edge Semantics]
    \label{as:partial-result}
    We say that we are under the optional edge semantics if, for any path function $f$, a sub-function of $f$ has exactly the same binding for its output variables as $f$.
\end{Definition}

\noindent The optional edge semantics mirrors the way real Web services work. Its main difference to the standard semantics is that it is not possible to use a function to filter out query results. For example, it is not possible to use the function \emph{get\-Company\-Info} to retrieve only those people who work at a company. The function will retrieve companies with addresses and companies without addresses, and we can find out the companies without addresses only by skimming through the results after the call. This contrasts with the standard semantics of parametrised queries (as used, e.g., in~\cite{romero2020equivalent,angie,susie}), which do not return a result if any of their variables cannot be bound. 

This has a very practical consequence: As we shall see, smart plans under the optional edge semantics have a very particular shape.

%% file: parts/smart-recognition-grammar-fabian.tex

\subsection{Preliminary Definitions}
\label{sec:preliminary-definitions}

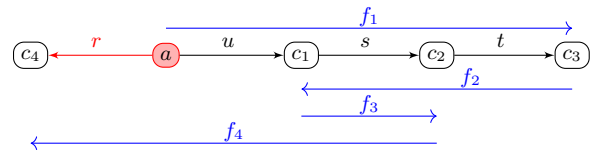
\begin{figure}
    \centering
    \input{images/CIKM2018/fig-plan-abstract}
    \caption{A bounded plan\label{fig:cikm-plan-abstract}}
\end{figure}

Our main intuition is that smart plans under the optional edge semantics walk forward until a turning point. From then on, they ``walk back'' to the input constant and query (see again Figure~\ref{fig:ex1}). As a more complex example, consider the atomic query $q(x) \leftarrow r(a, x)$ and the database shown in Figure~\ref{fig:cikm-plan-abstract}. The plan $f_1, f_2, f_3, f_4$ is shown in blue. 
As we can see, the plan walks ``forward'' and then ``backward'' again. 
Intuitively, the ``forward path'' makes sure that certain facts exist in the database (if the facts do not exist, the plan delivers no answer, and is thus trivially smart). If these facts exist, then all functions on the ``backward path'' are guaranteed to deliver results. Thus, if $a$ has an $r$-relation, the plan is guaranteed to deliver its object.
Let us now make this intuition more formal.

We first observe (and prove in
\ifthenelse{\boolean{final}}{the technical report)}{Property~\ref{prop:unconstraint-is-path-query})}
that the semantics of any filter-free execution plan can be reduced to a path query. The path query of Figure~\ref{fig:cikm-plan-abstract} is:
\begin{align*}
q(a,x) \leftarrow & u(a,y_1), s(y_1, y_2), t(y_2, y_3), t^-(y_3, y_2), s^-(y_2, y_1), \\
   & s(y_1, y_2), s^-(y_2, y_1), u^-(y_1, y_0), r(y_0, x)
\end{align*}
\noindent Now, any filter-free path query can be written unambiguously as the sequence of its relations -- the \emph{skeleton}. In the example, the skeleton is:
\[u.s.t.t^-.s^-.s.s^-.u^-.r\]
\noindent In particular, the skeleton of an atomic query $q(x) \leftarrow r(a, x)$ is just $r$. 
Given a skeleton $r_1.r_2...r_n$, we write $r_1...r_n(a)$ for the set of all answers of the query when $a$ is given as input. For path functions, we write the name of the function as a shorthand for the skeleton of the semantics of the function. For example, in Figure~\ref{fig:cikm-plan-abstract}, we have $f_1(a)=\{c_3 \}$, and $f_1f_2f_3f_4(a)=\{c_4 \}$. We now introduce two notions to formalise the ``forward and backward'' movement:

\begin{Definition}[Forward and Backward Step]
    Given a sequence of relations $r_0...r_n$ and a position $0\leq i\leq n$, a \emph{forward step} consists of the relation $r_i$, together with the updated position $i+1$.
    Given position $1\leq i\leq n+1$, a \emph{backward step} consists of the relation $r_{i-1}^-$, together with the updated position $i-1$. 
\end{Definition}


\begin{Definition}[Walk]
    A \emph{walk} to a position $k$ ($0 \leq k \leq n$) 
    through a sequence of relations $r_0...r_n$ consists of a sequence of steps (forward or backward) in $r_0...r_n$, so that the first step starts at position $n+1$, every step starts at the updated position of the previous step, and the last step leads to the updated position $k$.
\end{Definition}

\noindent If we do not mention $k$, we consider that $k=0$, i.e., we cross the sequence of relations entirely.

\begin{Example}
    In Figure~\ref{fig:cikm-plan-abstract}, a possible walk through $r^-ust$ is $t^-s^-ss^-u^-r$. This walk goes from $c_3$ to $c_2$ to $c_1$, then to $c_2$, and back through $c_1, c, c_4$ (as indicated by the blue arrows).
\end{Example}


\noindent We can now formalise the notion of the forward and backward path:
\begin{Definition}[Bounded plan]\label{def:bounded-plan}
    A bounded path for a set of relations $\mathcal{R}$ and a query $q(x) \leftarrow r(a, x)$ is a path query $P$, followed by a walk through $r^-P$.
    A bounded plan for a set of path functions $\mathcal{F}$ is a non-redundant execution plan whose semantics are a bounded path.
    We call $P$ the \textit{forward path} and the walk though $r^-P$ the \textit{backward path}.
\end{Definition}

\begin{Example}
    In Figure~\ref{fig:cikm-plan-abstract}, $f_1f_2f_3f_4$ is a bounded path, where the forward path is $f_1$, and the backward path $f_2f_3f_4$ is a walk through $r^-f_1$.
\end{Example}

\subsection{Characterising Smart Plans}
\label{sec:charac-smart-plan}

Our notion of bounded plans is based purely on the notion of skeletons, and does not make use of filters. This is not a problem, because smart plans depend on constraint-free plans. Furthermore, we show in
\ifthenelse{\boolean{final}}{the technical report}{Appendix~\ref{sec:restrict-path-queries}}
that we can restrict ourselves to execution plans whose semantics is a path query. 
This allows for the following theorems (proven in
\ifthenelse{\boolean{final}}{the technical report)}{the appendix)}:

\begin{Theorem}[Correctness]\label{thm:bounded-correct}
    Let $q(x) \leftarrow r(a, x)$ be an atomic query, $F$ a set of path functions and $F_{sub}$ the set of sub-functions of $F$. Let $\pi_a$ be a non-redundant bounded execution plan over the $F_{sub}$ such that its semantics is a path query. Then $\pi_a$ is weakly smart.
\end{Theorem}

\begin{Theorem}[Completeness]\label{thm:bounded-complete}
    Let $q(x) \leftarrow r(a, x)$ be an atomic query, $F$ a set of path functions and $F_{sub}$ the set of sub-functions of $F$. Let $\pi_a$ be a weakly smart plan over $F_{sub}$ such that its semantics is a path query. Then $\pi_a$ is bounded.
\end{Theorem}

\noindent We have thus found a way to recognise weakly smart plans without executing them. Extending this characterisation from weakly smart plans to fully smart plans consists mainly of adding a filter.
\ifthenelse{\boolean{final}}{The technical report gives more technical details.}{Appendix~\ref{sec:generate-smart-plans} shows how the adaptation is done in practice and Appendix~\ref{sec:characterising-smart-plans} gives more technical details.}

%% file: images/CIKM2018/fig-plan-abstract.tex
\tikzstyle arrowstyle=[blue,semitransparent,scale=2]
\tikzstyle basiclabel=[draw=none,fill=none,shape=rectangle,inner sep=2pt,scale=.8]
\tikzstyle leftlabel=[basiclabel,anchor=east]
\tikzstyle rightlabel=[basiclabel,anchor=west]
\tikzstyle bottomlabel=[basiclabel,anchor=north]
\tikzstyle toplabel=[basiclabel,anchor=south]

\definecolor{lightblue}{cmyk}{0.12,0,0,0.15} 
\tikzstyle{block} = [rectangle, draw, fill=white, 
     text centered, rounded corners]
     
\tikzstyle{headvar} = [rectangle, draw, fill=lightblue!70,
     text centered, rounded corners]    
\tikzstyle{labelcase} = [rectangle,    text centered]
\tikzstyle{background}=[rectangle,   fill=lightblue!70,
                                                inner sep=0.2cm,
                                                rounded corners=5mm]

\tikzstyle{line} = [draw, -latex']
\tikzstyle{linenoarrow} = [draw]
\tikzstyle{invisibleline} = [-latex',sloped]
\tikzstyle{dashedline} = [draw, dashed]
\tikzstyle{inputtar} = [rectangle, draw, fill=lightblue!70, 
     text centered, rounded corners]

\tikzstyle{inputvar} = [rectangle, draw=red,  fill=red!30,
text centered, rounded corners]

\tikzstyle{textnode} = [rectangle, text centered]

\begin{tikzpicture}[scale=0.9, transform shape]


\node [inputvar] (x3) {$a$};
\node [block] (y3) at ($(x3)+(+2cm,0cm)$) {$c_1$};
\node [block] (z3) at ($(y3)+(+2cm,0cm)$) {$c_2$};
\node [block] (t3) at ($(z3)+(+2cm,0cm)$) {$c_3$};
\node [block] (u3) at ($(x3)+(-2cm,0cm)$) {$c_4$};
\path [line] (x3) --node [above,align=center] {{\textit{$u$}} } (y3);
\path [line] (y3) --node [above,align=center] {{\textit{$s$}} } (z3);
\path [line] (z3) --node [above,align=center] {{\textit{$t$}} } (t3);
\path [red, line] (x3) --node [above,align=center] {{\textit{$r$}} } (u3);

\coordinate (xu3) at   ($(x3) + (+0cm,+0.4cm) $); 
\coordinate (tu3) at   ($(t3) + (+0cm,+0.4cm) $); 
\draw [blue,->]     (xu3) to  (tu3);
\node [labelcase,blue] (lf1) at   ($0.5*(xu3) + 0.5*(tu3) + (0,0.17cm)$) {\textit{$f_1$}};

\coordinate (tb3f2) at   ($(t3) + (+0cm,-0.5cm) $); 
\coordinate (yb3f2) at   ($(y3) + (+0cm,-0.5cm) $); 
\draw [blue,->]     (tb3f2) to  (yb3f2);
\node [labelcase,blue] (lf1) at   ($0.5*(tb3f2) + 0.5*(yb3f2) + (0.5cm,+0.17cm)$) {\textit{$f_2$}};

\coordinate (zb3f3) at   ($(z3) + (+0cm,-0.9cm) $); 
\coordinate (yb3f3) at   ($(y3) + (+0cm,-0.9cm) $); 
\draw [blue,->]     (yb3f3) to  (zb3f3);
\node [labelcase,blue] (lf1) at   ($0.5*(zb3f3) + 0.5*(yb3f3) + (0cm,+0.17cm)$) {\textit{$f_3$}};

\coordinate (zb3f4) at   ($(z3) + (+0cm,-1.3cm) $); 
\coordinate (ub3f4) at   ($(u3) + (+0cm,-1.3cm) $); 
\draw [blue,->]     (zb3f4) to  (ub3f4);
\node [labelcase,blue] (lf1) at   ($0.5*(zb3f4) + 0.5*(ub3f4) + (0cm,+0.17cm)$) {\textit{$f_4$}};

\end{tikzpicture}

%% file: parts/algorithm.tex
\section{Generating Smart Plans}\label{sec:generation-main}

We have shown that weakly smart plans are precisely the bounded plans. We will now turn to generating such plans. 
Let us first introduce the notion of minimal plans.

\subsection{Minimal Smart Plans}

In line with related work~\cite{romero2020equivalent}, we will not generate redundant plans. These contain more function calls, 
and cannot deliver more results than non-redundant plans. More precisely, we will focus on \textit{minimal plans}:

\begin{Definition}[Minimal Smart Plan]\label{minimal}
Let $\pi_a(x)$ be a non-redundant execution plan organised in a sequence $c_0, c_1,  \ldots, c_k$ of calls, such that the input of $c_0$ is the constant $a$, every other call $c_i$ takes as input an output variable of the previous call $c_{i-1}$, and the output of the plan is in the call $c_k$.
$\pi_a$ is a \emph{minimal (weakly) smart plan} if it is a (weakly) smart plan and there exists no other (weakly) smart plan $\pi'_a(x)$ composed of a sub-sequence $c_{i_1},..., c_{i_n}$ (with $0 \leq i_1 < ... < i_n \leq k$).
\end{Definition}

\begin{Example}
Let us consider the two functions $f_1(x, y) = r(x, y)$ and $f_2(y, z) = r^-(y, t).r(t, z)$. For the query $q(x) \leftarrow r(a, x)$, the plan $\pi_a(x) = f_1(a, y), f_2(y, x)$ is obviously weakly smart. It is also non-redundant. However, it is not minimal. This is because $\pi_a'(x) = f_1(a, x)$ is also weakly smart, and is composed of a sub-sequence of calls of $\pi_a$.
\end{Example}

\noindent In general, it is not useful to consider non-minimal plans because they are just longer but cannot yield more results. On the contrary, a non-minimal plan can have fewer results than its minimal version, because the additional calls can filter out results.
The notion of minimality would make sense also in the case of equivalent rewritings. However, in that case, the notion would impact just the number of function calls and not the results of the plan, since equivalent rewritings deliver the same results by definition. 
In the case of smart plans, as we will see, the notion of minimality allows us to consider only a finite number of execution plans and thus to have an algorithm that terminates.

\subsection{Bounding and Generating the Weakly Smart Plans}
\label{bound_and_generate}

We can enumerate all minimal weakly smart plans because their number is limited. We show in
\ifthenelse{\boolean{final}}{the technical report}{Appendix~\ref{sec-bound}}
the following theorem:

\begin{Theorem}[Bound on Plans]\label{bound-plan-main-text}
    Given a set of relations $\mathcal{R}$, a query $q(x) \leftarrow r(a, x)$, $r \in \mathcal{R}$, and a set of path function definitions $\mathcal{F}$, there can be no more than $M!$ minimal smart plans, where $M={|\mathcal{F}|}^{2k}$ and $k$ is the maximal number of atoms in a function. Besides, there exists an algorithm to enumerate all minimal smart plans.
\end{Theorem}

\noindent This bound is very pessimistic: In practice, the plans are very constrained and thus, the complete exploration is quite fast, as we will show in Section~\ref{sec:cikm-experiments}.

The intuition of the theorem is as follows: Let us consider a bounded path with a forward and a backward path. For each position $i$, we consider a state that represents the functions crossing the position $i$ (we also consider function starting and ending there). We notice that, as the plan is minimal, there cannot be two functions starting at position $i$ (otherwise the calls between these functions would be useless). This fact limits the size of the state to $2k$ functions (where $k$ is the maximal size of a function, the $2$ is due to the existence of both a forward and backward path). Finally, we notice that a state cannot appear at two different positions; otherwise, the plan would not be minimal (all function calls between the repetition are useless). Thus, the algorithm we propose explores the space of states in a finite time, and yields all minimal smart plans. At each step of the search, we explore the adjacent nodes that are consistent with the current state. In practice, these transitions are very constrained, and so the complexity is rarely exponential (as we will see in the experiments).

%% file: appendix/smart-generation.tex

\subsection{Generating the Weakly Smart Plans}\label{sec-algorithm}

\noindent Theorem~\ref{bound-plan-main-text} allows us to devise an algorithm that enumerates all minimal weakly smart plans. For simplicity, let us first assume that no function definition contains a loop, i.e., no function contains two consecutive relations of the form $rr^-$. This means that a function cannot be both on a forward and backward direction. We will later see how to remove this assumption. Algorithm~\ref{algo:cikm-algorithm} takes as input a query $q$ and a set of function definitions $\mathcal{F}$. It first checks whether the query can be answered trivially by a single function (Line~1). If that is the case, the plan is printed (Line~2). Then, the algorithm sets out to find more complex plans. To avoid exploring states twice, it keeps a history of the explored states in a stack $H$ (Line~3). The algorithm finds all non-trivial functions $f$ that could be used to answer $q$. These are the functions whose short notation ends in $q$ (Line~4). For each of these functions, the algorithm considers all possible functions $f'$ that could start the plan (Line~5). For this, $f'$ has to be \emph{consistent} with $f$, i.e., the functions have to share the same relations. The pair of $f$ and $f'$ constitute the first state of the plan. Our algorithm then starts a depth-first search from that first state (Line~6). For this purpose, it calls the function \emph{search} with the current state, the state history, and the set of functions. In the current state, a marker (a star) designates the forward path function.

\begin{algorithm}\caption{FindMinimalWeakSmartPlans}\label{algo:cikm-algorithm}
    \DontPrintSemicolon
    \KwData{Query $q(a) \leftarrow r(a, x)$, set of path function definitions and all their sub-functions $\mathcal{F}$}
    \KwResult{Prints minimal weakly smart plans}
    \If{$\exists f=r \in \mathcal{F}$}{
        print($f$)
    }
    $H\gets Stack()$\;
    \ForEach{$f=r_1...r_n.r \in \mathcal{F}$}{
        \ForEach{$f' \in \mathcal{F}$ consistent with $r_n^-...r_1^-$}{
            search($\{\langle f,n,backward\rangle, \langle f',1,forward\rangle^*\}$, $H$, $\mathcal{F}$)
        }    
    }    
\end{algorithm}

Algorithm~\ref{search} executes a depth-first search on the space of states. It first checks whether the current state has already been explored (Line~1). If that is the case, the method just returns. Otherwise, the algorithm creates the new state $S'$ (Line~3). For this purpose, it considers all positioned functions in the forward direction (Lines~5-7). If any of these functions ends, the end counter is increased (Line~6). Otherwise, we advance the positioned function by one position. The $(*)$ means that if the positioned function happens to be the designated forward path function, then the advanced  positioned function has to be marked as such, too. We then apply the procedure to the backwards-pointing functions (Lines~8-11).

Once that is done, there are several cases: If all functions ended, we have a plan (Line~12). In that case, we can stop exploring because no minimal plan can include an existing plan. Next, the algorithm considers the case where one function ended, and one function started (Line~13). If the function that ended were the designated forward path function, then we would have to add one more forward function. However, then the plan would contain two functions that start at the current state. Since this is not permitted, we just do not do anything (Line~14), and the execution jumps to Line~29. If the function that ended was some other function, then the ending and the starting function can form part of a valid plan. No other function can start or end at the current state, and hence we just move to the next state (Line~15). Next, the algorithm considers the case where one function starts and no function ends (Line~16). In that case, it has to add another backward function. It tries out all functions (Line~17-19) and checks whether adding the function to the current state is consistent (as in Algorithm~\ref{algo:cikm-algorithm}). If that is the case, the algorithm calls itself recursively with the new state (Line~19). Lines~20-23 do the same for a function that ended. Here again, the $(*)$ means that if $f$ was the designated forward path function, then the new function has to be marked as such. Finally, the algorithm considers the case where no function ended, and no function started (Line~24). In that case, we can just move on to the next state (Line~25). We can also add a pair of a starting function and an ending function. Lines~26-28 try out all possible combinations of a starting function and an ending function and call the method recursively. If none of the previous cases applies, then $end>1$ and $start>1$. This means that the current plan cannot be minimal. In that case, the method pops the current state from the stack (Line~29) and returns.
    
\begin{algorithm}\caption{Search}\label{search}
    \DontPrintSemicolon
    \KwData{A state $S$ with a designated forward path function, a set of states $\mathcal{H}$, a set of path functions $\mathcal{F}$}
    \KwResult{Prints minimal weakly smart plans}
    \textbf{if} $S\in H$ \textbf{then} return\;
    $H.push(S)$\;
    $S'\gets\emptyset$\;
     $end\gets 0$\;
    \ForEach{$\langle r_1...r_n,i,forward\rangle \in S$}{
        \textbf{if} $i+1>n$ \textbf{then} $end++$\;
        \textbf{else} $S'\gets S'\cup\{\langle r_1...r_n,i+1,forward\rangle^{(*)}\}$
    }
    $start\gets 0$\;
    \ForEach{$\langle r_1...r_n,i,backward\rangle \in S$}{
        \textbf{if} $i=1$ \textbf{then} $start++$\;
       \textbf{else} $S'\gets S'\cup\{\langle r_1...r_n,i-1,backward\rangle\}$
    }
    \textbf{if} $S'=\emptyset$ \textbf{then} print($H$)\;
    \ElseIf{$start=1 \wedge end=1$}{
        \textbf{if} the designated function ended \textbf{then} pass\;
        \textbf{else} search($S', H, \mathcal{F}$)\;
    }
    \ElseIf{$start=1 \wedge end=0$}{
        \ForEach{$f \in \mathcal{F}$}{
            $S''\gets S' \cup \{\langle f,|f|,backward\rangle\}$\;
            \textbf{if} $S''$ is consistent \textbf{then} search($S'',H,\mathcal{F}$)
        }
    }
    \ElseIf{$start=0 \wedge end=1$}{
        \ForEach{$f \in \mathcal{F}$}{
            $S''\gets S' \cup \{\langle f,1,forward\rangle^{(*)}\}$\;
            \textbf{if} $S''$ is consistent \textbf{then} search($S'',H,\mathcal{F}$)
        }
    }
    \ElseIf{$start=0 \wedge end=0$}{
        search($S', H, \mathcal{F}$)\;
        \ForEach{$f, f'\in\mathcal{F}$}{
            $S''\gets S' \cup \{\langle f,1,forward\rangle, \langle f',|f'|,backward\rangle\}$\;
            \textbf{if} $S''$ is consistent \textbf{then} search($S'',H,\mathcal{F}$)
        }
    }
    $H.pop()$
\end{algorithm}

\begin{Theorem}[Algorithm]\label{theorem-algorithm}
    Algorithm~\ref{algo:cikm-algorithm} is correct and complete, terminates on all inputs, and runs in time $\mathcal{O}(M!)$, where $M={|\mathcal{F}|}^{2k}$ and $k$ is the maximal number of atoms in a function.
\end{Theorem}

\noindent The worst-case runtime of $\mathcal{O}(M!)$ is unlikely to appear in practice. Indeed, the number of possible functions that we can append to the current state in Lines 19, 23, 28 is severely reduced by the constraint that they must coincide on their relations with the functions that are already in the state. In practice, very few functions have this property. Furthermore, we can significantly improve the bound if we are interested in finding only a single weakly smart plan:

\begin{Theorem}\label{find-single}
Given an atomic query and a set of path function definitions $\mathcal{F}$, we can find a single weakly smart plan in $\mathcal{O}({|\mathcal{F}|}^{2k})$, where $k$ is the maximal number of atoms in a function.
\end{Theorem}

\paragraph{Functions with loops} If there is a function that contains a loop of the form $r.r^-$, then Algorithm~\ref{search} has to be adapted as follows: First, when neither functions are starting nor ending (Lines~24-28), we can also add a function that contains a loop. Let $f = r_1...r_i r_i^-... r_n$ be such a function. Then the first part $r_1 ... r_i$ becomes the backward path, and the second part $r_i^- ... r_n$ becomes the forward path in Line~27.

When a function ends (Lines~20-23), we could also add a function with a loop. Let $f = r_1...r_i r_i^- r_n$ be such a function. The first part $r_1 ... r_i$ will create a forward state $\langle r_1 ... r_i,1,\textit{forward}\rangle$. The second part, $r_i^- ... r_n$ will create the backward state $\langle r_i^- ... r_n,|r_1 ... r_i|,\textit{backward}\rangle$. The consistency check has to be adapted accordingly. The case when a function starts (Lines~16-19) is handled analogously. Theorems~\ref{theorem-algorithm} and~\ref{find-single} remain valid, because the overall number of states is still bounded as before.

%% file: parts/experiments.tex
\section{Experiments}\label{sec:cikm-experiments}

We have implemented the Susie Algorithm~\cite{susie} (more details in
\ifthenelse{\boolean{final}}{the technical report}{Appendix~\ref{Susie}}), the equivalent rewriting approach~\cite{romero2020equivalent}, as well as our method (Section~\ref{bound_and_generate}) in Python. The code is available on Github\footnote{\href{https://github.com/Aunsiels/smart_plans}{https://github.com/Aunsiels/smart\_plans}}. We conduct two series of experiments -- one on synthetic data, and one on real Web services. All our experiments are run on a laptop with Linux, 1 CPU with four cores at 2.5GHz, and 16 GB RAM.

\subsection{Synthetic Functions}

In our first set of experiments, we use the methodology introduced by~\cite{romero2020equivalent} to simulate random functions. We consider a set of artificial relations $\mathcal{R}=\{r_1,...,r_n\}$, and randomly generated path functions up to length 3, where all variables are existential except the last one. Then we try to find a smart plan for each query of the form $q(x) \leftarrow r(a, x), r\in\mathcal{R}$. 

In our first experiment, we limit the number of functions to 30 and vary the number $n$ of relations. All the algorithms run in less than 2 seconds in each setting for each query. Figure~\ref{fig:answerednrelations} shows which percentage of the queries the algorithms answer. As expected, when increasing the number of relations, the percentage of answered queries decreases, as it becomes harder to combine functions. The difference between the curve for weakly smart plans and the curve for smart plans shows that it was not always possible to filter the results to get exactly the answer of the query. Weakly smart plans can answer more queries but at the expense of delivering only a super-set of the query answers. In general, we observe that our approach can always answer strictly more queries than Susie and the equivalent rewriting approach. 

\begin{figure}
\centering
\begin{subfigure}{.5\textwidth}
  \centering
  \includegraphics[width=\columnwidth]{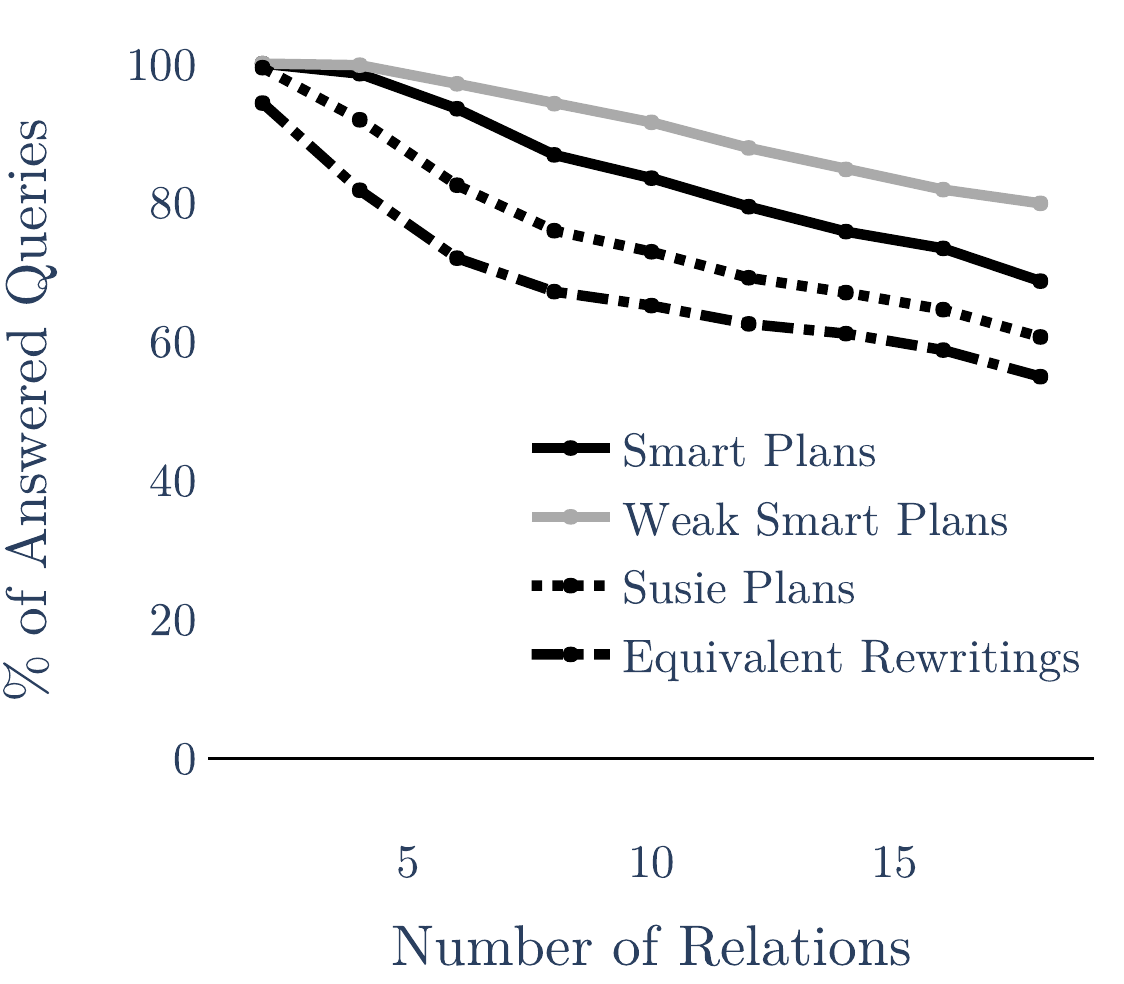}
  \caption{}\label{fig:answerednrelations}
\end{subfigure}%
\begin{subfigure}{.5\textwidth}
  \centering
  \includegraphics[width=\columnwidth]{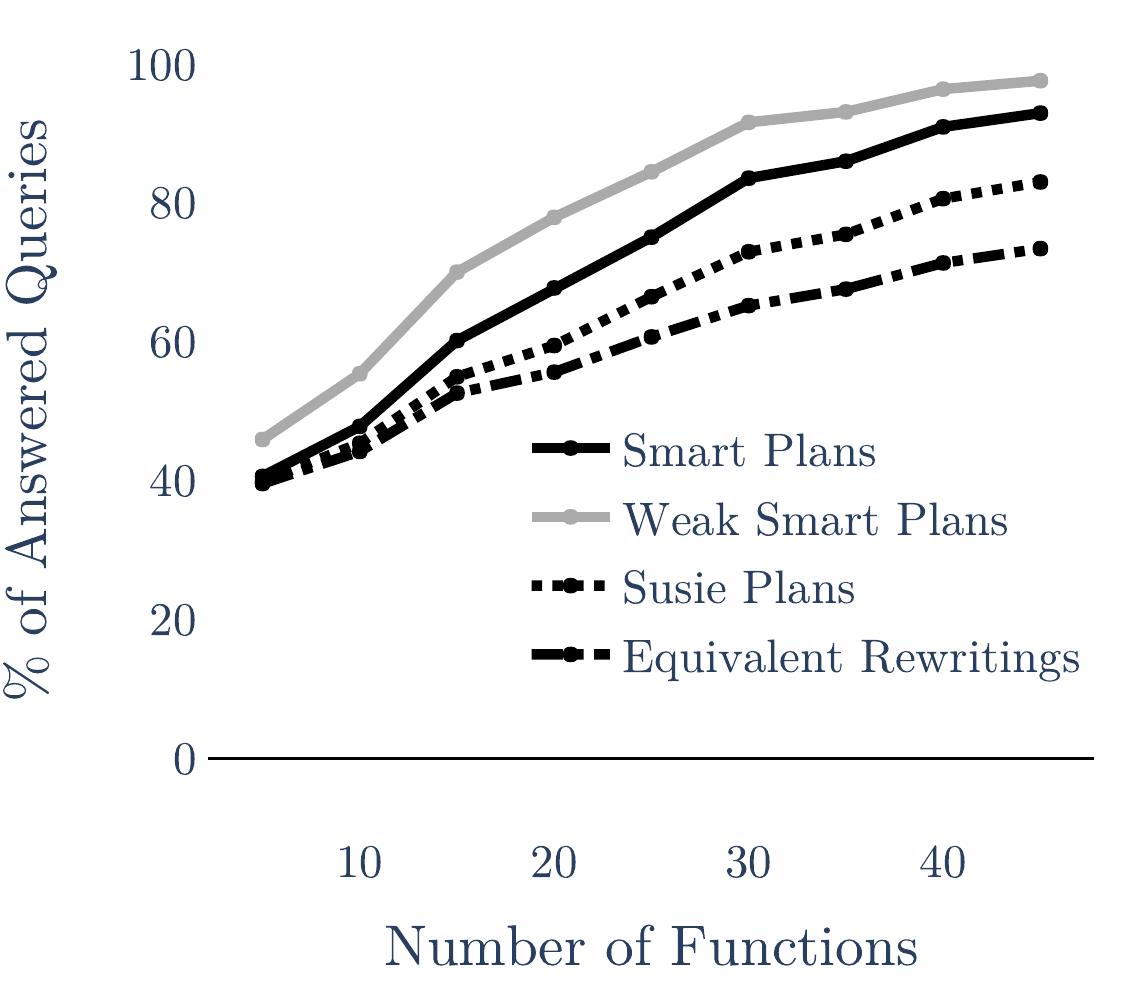}
  \caption{}\label{fig:answerednfunctions}
\end{subfigure}
\caption{Percentage of answered queries}
\end{figure}

In our next experiment, we fix the number of relations to 10 and vary the number of functions. Figure \ref{fig:answerednfunctions} shows the results. As we increase the number of functions, we increase the number of possible function combinations. Therefore, the percentage of answered queries increases for all approaches. As before, our algorithm outperforms the other methods by a wide margin. The reason is that Susie cannot find all smart plans (see \ifthenelse{\boolean{final}}{the technical report}{Appendix~\ref{Susie}} for more details). Equivalent rewritings, on the other hand, can find only those plans that are equivalent to the query on all databases -- which are very few in the absence of constraints.

\subsection{Real-World Web Services}
In our second series of experiments, we apply the methods to real-world Web services. We use the functions provided by~\cite{romero2020equivalent,susie}. These are the functions of the Web services of \href{http://abebooks.com}{Abe Books}, \href{http://isbndb.com/}{ISBNDB}, \href{http://www.librarything.com/}{LibraryThing}, \href{http://musicbrainz.org/}{MusicBrainz}, and \href{https://www.themoviedb.org}{MovieDB}. Besides, as these Web services do not contain many existential variables, we added the set of functions based on information extraction techniques (IE) from~\cite{susie}. 

Table~\ref{table:completenessTable} shows the number of functions and the number of relations for each Web service. Table~\ref{table:examplesFunctions} gives examples of functions. Some of them are recursive. For example, MusicBrainz allows querying for the albums that are related to a given album. All functions are given in the same schema. Hence, in an additional setting, we consider the union of all functions from all Web services. 

Note that our goal is not to call the functions. Instead, our goal is to determine whether a smart plan exists -- before any functions have to be called.

\begin{table}[t]
\centering
\begin{tabular}{l}
  \toprule  
$\textit{getDeathDate}(\underline{x},y,z) \leftarrow \textit{hasId}^{-}(x,y) \wedge \textit{diedOnDate}(y,z)$ \\
  
$\textit{getSinger}(\underline{x},y,z,t) \leftarrow \textit{hasRelease}^{-}(x,y) \wedge \textit{released}^{-}(y,z) \wedge \textit{hasId}(z,t)$ \\
  
$\textit{getLanguage}(\underline{x},y,z,t) \leftarrow \textit{hasId}(x,y)\wedge \textit{released}(y,z) \wedge \textit{language}(z,t)$ \\
  
 $\textit{getTitles}(\underline{x},y,z,t) \leftarrow \textit{hasId}^{-}(x,y) \wedge \textit{wrote}^-(y,z) \wedge title(z,t)$ \\
  
 $\textit{getPublicationDate}(\underline{x},y,z) \leftarrow \textit{hasIsbn}^{-}(x,y) \wedge \textit{publishedOnDate}(y,z)$ \\
  \bottomrule
\end{tabular}
\caption{Examples of real functions (3 of MusicBrainz, 1 of ISBNdb, 1 of LibraryThing).}\label{table:examplesFunctions}
\end{table}

For each Web service, we considered all queries of the form $q(x) \leftarrow r(a,x)$ and $q(x) \leftarrow r^-(a,x)$, where $r$ is a relation used in the function definitions of that Web service. We ran the Susie algorithm, the equivalent rewriting algorithm, and our algorithm for each of these queries. The run-time is always less than 2 seconds for each query. Table~\ref{table:completenessTable} shows the ratio of queries for which we could find smart plans. We first observe that our approach can always answer at least as many queries as the other approaches can answer. Furthermore, there are cases where our approach can answer strictly more queries than Susie.

\begin{table}[t]
\centering
\begin{tabular}{l|cc|r|r|r|r}
  \toprule
  Web Service & Functions & Relations & Susie & Eq. Rewritings & Smart Plans\\
  \midrule
  MusicBrainz (+IE) & 23 & 42 & 48\% (32\%) & 48\% (32\%) & 48\%  (35\%)\\
  LastFM (+IE) & 17 & 30 & 50\% (30\%) & 50\% (30\%) & 50\% (32\%) \\
  LibraryThing (+IE) & 19 & 32 & 44\% (27\%) & 44\% (27\%) & 44\% (35\%) \\
  Abe Books (+IE) & 9 & 8 & 75\% (14\%) & 63\% (11\%) & 75\% (14\%) \\
  ISBNdb (+IE) & 14 & 20 & 65\% (23\%) & 50\% (18\%) & 65\% (23\%) \\
  Movie DB (+IE) & 12 & 18 & 56\% (19\%) & 56\% (19\%) & 56\% (19\%)\\
  UNION with IE & 74 & 82 & 52\% & 50\% & 54\%\\
  \bottomrule 
\end{tabular}
\caption{Percentage of queries with smart plans (numbers in parenthesis represent the results with IE).}\label{table:completenessTable}
\end{table}

\noindent \textbf{The advantage of our algorithm is not that it beats Susie by some percentage points on some Web services. Instead, the crucial advantage of our algorithm is the guarantee that the results are complete.} If our algorithm does not find a plan for a given query, it means that there cannot exist a smart plan for that query. Thus, even if Susie and our algorithm can answer the same number of queries on AbeBooks, only our algorithm can guarantee that the other queries cannot be answered at all. Thus, only our algorithm gives a complete description of the possible queries of a Web service.

Rather short execution plans can answer some queries. Table~\ref{eg:cikm-examples} shows a few examples. However, a substantial percentage of queries cannot be answered at all. In MusicBrainz, for example, it is not possible to answer \emph{produced}$(a,x)$ (i.e., to know which albums a producer produced), \emph{hasChild$^-$(a,x)} (to know the parents of a person), and \emph{marriedOnDate$^-(a,x)$} (to know who got married on a given day). 
These observations show that the Web services maintain control over the data, and do not allow exhaustive requests.

\begin{table}[t]
    \centering
    \begin{tabular}{lll}
        \toprule
        Query & Plan\\
        \midrule
        \emph{hasTrackNumber} & \emph{getReleaseInfoByTitle, getReleaseInfoById}\\  
        \emph{hasIdCollaborator} & \emph{getArtistInfoByName, getCollaboratorIdbyId,getCollaboratorsById}\\  
        \emph{publishedByTitle} & \emph{getBookInfoByTitle, getBookInfoById}\\
        \bottomrule
    \end{tabular}
    \caption{Example Plans (2 of MusicBrainz, 1 of ABEBooks).}\label{eg:cikm-examples}
\end{table}

\balance

%% file: parts/conclusion.tex
\section{Conclusion}

In this paper, we have introduced the concept of smart execution plans for Web service functions. These are plans that are guaranteed to deliver the answers to the query if they deliver results at all. We have formalised the notion of smart plans, and we have given a correct and complete algorithm to compute smart plans. Our experiments have demonstrated that our approach can be applied to real-world Web services. All experimental data, as well as all code, is available at the URL given in Section~\ref{sec:cikm-experiments}. We hope that our work can help Web service providers to design their functions, and users to query the services more efficiently.\\

%% file: appendix/appendix.tex
\input{appendix/susie.tex}

\input{appendix/bound_weak_smart_plans.tex}

\input{appendix/only_minimal.tex}

\input{appendix/smart_generation_adaptation.tex}

\section{Proof of the Properties And Theorems}

\subsection{Why we can restrict to path queries}
\label{sec:restrict-path-queries}

Under the optional edge semantics, we can query for all partial results. Now if we have all sub-functions of a given function $f$, then $f$ itself is no longer necessary (except if it is itself a sub-function). We call this substitution by sub-functions the sub-function transformation:

\begin{Definition}[Sub-Function Transformation]
    \label{def:sub-function-transformation}
    Let $\mathcal{F}$ be a set of path functions and $\mathcal{F}_{sub}$ the set of sub-functions of $\mathcal{F}$. Let $\pi_a$ be a non-redundant execution plan over $\mathcal{F}$. Then, we define the \textit{Sub-Function Transformation} of $\pi_a$, written $\mathcal{P}(\pi_a)$, as the non-redundant execution plan over the sub-functions $\mathcal{F}_{sub}$ as follows:
    \begin{itemize}
        \item The output is the same than $\pi_a$
        \item Each function call $c$ in $\pi_a$ associated to a path function $f$ is replaced by the smallest sub-function of $f$ which contains the output variables which are either the output of the plan, used by other calls or involved in filters.
    \end{itemize}
\end{Definition}

\noindent This transformation has the property to conserve the smartness.

\begin{Property}
    \label{prop:sub-function-transformation-safe}
    Let $q(x) \leftarrow r(a, x)$ be an atomic query and $\mathcal{F}$ be a set of path functions. Let $\pi_a$ be a non-redundant execution plan composed of the sequence of calls $c_1,..., c_n$. Then, under the optional edge semantics, $\pi_a$ is smart (resp. weakly smart) iff its sub-function transformation $\mathcal{P}(\pi_a)$ is smart (resp. weakly smart).
\end{Property}

The property tells us that under the \textit{optional edge semantics}, we can replace all path functions by all of their sub-functions.

Finally, in the case of constraint-free plans, we have that the sub-function transformation creates a path query, which will be easier to manipulate.

\begin{Property}
    \label{prop:unconstraint-is-path-query}
    Let $\pi_a$ be a constraint-free non-redundant execution plan. Then the semantics of the sub-function transformation $\mathcal{P}(\pi_a)$ is a path query where the last variable of the last atom is the output atom.
\end{Property}

We can also deduce from this property that once we have transformed a constraint-free execution plan (exploited in weak smart plans) to use only sub-functions, we can write the semantics of the plan unambiguously as a skeleton. In particular, we could consider that each sub-function has only one output.

In the end, we can see that it is safe to consider only execution plans whose semantics is a path query in the case of constraint-free plans. We shall see that it is also the case for minimal-filtering non-redundant execution plans.

\subsection{Proof of Property~\ref{prop:sub-function-transformation-safe}}

This property derives directly from the optional edge semantics: The outputs of the sub-function are precisely the same than the outputs of the full function for the same variables. So, we can apply all the filters and have the same effect.

\subsection{Proof of Property~\ref{prop:unconstraint-is-path-query}}

This property follows directly the fact that in a constraint-free plan, we require only one output per function. Therefore, function calls can be chained, and the last variable of the last atom is the output of the plan (we require nothing else after that).

\subsection{Proof of Theorem~\ref{thm:bounded-correct}}

As $\pi_a$ is a constraint-free non-redundant execution plan over the $F_{sub}$, its sub-function transformation can be written as a path query (see Property~\ref{prop:unconstraint-is-path-query}). We now consider we have this form.

Let $\pi_a$ be a bounded plan for a query $q(x) \leftarrow r(a, x)$.
Assume a database $\mathcal{I}$ such that $q(\mathcal{I}) \neq \emptyset$ and $\pi_a(\mathcal{I}) \neq \emptyset$ (such a database exists).
Choose any constant $c_0\in q(\mathcal{I})$. We have to show that $c_0 \in \pi_a(\mathcal{I})$. Since $\pi_a$ is bounded, its consequences can be split into a forward path $F=r_1...r_m$ and a following backward path $B=r'_1...r'_n$ with $r'_n=r$ (by definition of a walk).
Since $\pi_a(\mathcal{I}) \neq \emptyset$ it follows that $F(a) \neq \emptyset$. Hence, $\mathcal{I}$ must contain $r_1(a,c_2)...r_m(c_m,c_{m+1})$ for some constants $c_1,...,c_{m+1}$.
Since $r(a, c_0) \in \mathcal{I}$, the database $\mathcal{I}$ must contain $r^-(c_0,a) r_1(a,c_2)...r_m(c_m,c_{m+1})$. 
$B=r'_1...r'_n$ is a walk through $r^-F$. Let us prove by induction that $c_i \in Fr_1'...r_j'(a)$ if $r_j'$ was generated by a step that leads to position $i$. To simplify the proof, we call $c_1 = a$.

The walk starts at position $i=m+1$ with a backward step, leading to position $i=m$. Hence, $r'_1=r_m^-$. Thus, we have $c_{m}\in Fr'_1(a)$.
Now assume that we have arrived at some position $i\in[1,m]$, and that $c_i\in Fr_1'...r_j'(a)$, where $r_j'$ was generated by the previous step. If the next step is a forward step, then $r_{j+1}=r_i$, and the updated position is $i+1$. Hence, $c_{i+1}\in Fr_1'...r_{j+1}'(a)$. If the next step is a backward step, then $r_{j+1}=r_{i-1}^-$, and the updated position is $i-1$. Hence, $c_{i-1}\in Fr_1'...r_{j+1}'(a)$. It follows then that $c_0\in FB(a)$ as the walk ends at position 0.

\subsection{Proof of Theorem~\ref{thm:bounded-complete}}
\label{sec:proof-weak-smart-complete}

As $\pi_a$ is a constraint-free non-redundant execution plan over the $F_{sub}$, its sub-function transformation can be written as a path query (see Property~\ref{prop:unconstraint-is-path-query}). We now consider we have this form.

Let $\pi_a$ be a weak smart plan for a query $q(x) \leftarrow r(a, x)$, with consequences $r_1(x_1,x_2)...$ $r_n(x_n,x_{n+1}),$  $r_{n+1}(x_{n+1},x_{n+2})$. Without loss of generality, we suppose that $r_1 \neq r$ (the proof can be adapted to work also in this case). Consider the database $\mathcal{I}$
    \[\mathcal{I}=\{r^-(c_0,a), r_1(a,c_2), ..., r_n(c_n,c_{n+1}), r_{n+1}(c_{n+1},c_{n+2})\}\]
Here, the $c_i$ are fresh constants. For convenience, we might write $a = c_1$. On this database, $\pi_a(\mathcal{I})\supseteq \{c_{n+2}\}\neq\emptyset$ and $q(\mathcal{I})=\{c_0\}$.
Since $\pi_a$ is weakly smart, we must have $c_0\in \pi_a(\mathcal{I})$.
Let $\sigma$ be a binding for the variables of $\pi_a$ that produces this result $c_0$, i.e., $\sigma(x_1)=\sigma(x_{n+1})=a$ (as $a$ is the only entity linked to $c_0$).
We notice that we must have $r_{n + 1} = r$ as it is the only relation which leads to $c_0$. Let us define \[m=(\textit{max}~~\{ m'~|~\exists l: \sigma(x_l)=c_{m'}\})-1\]

Let us call $r^-r_1...r_m$ the \emph{forward path}, and $r_{m+1}...r_nr_{n+1}$ the \emph{backward path} (with $r_{n+1}=r$). We have to show that the backward path is a walk in the forward path. 

Let us show by induction that $r_{m+1}...r_j$ can be generated by $j-m$ steps (with $j\in[m+1,n+1])$, so that the updated position after the last step is $i$, and $\sigma(x_{j+1})=c_i$. We first note that, due to the path shape of $\mathcal{I}$, $\sigma(x_{j})=c_i$  for any $i,j$ always implies that $\sigma(x_{j+1})\in\{c_{i-1},c_{i+1}\}$.
Let us start with $j=m+1$ and position $i=m+1$. Since $\sigma(x_{j+1})$ cannot be $c_{m+1}$ (by definition of $m$), we must have $\sigma(x_{j+1})=c_{m}$.
Since $\sigma(x_j)=c_{m+1}$ and $\sigma(x_{j+1})=c_{m}$, we must have $r_{m+1}=r_{m}^-$. Thus, $r_{m+1}$ was generated by a backward step, and we call $i-1$ the updated position.
Now let us assume that we have generated $r_{m+1}...r_j$ (with $j\in[m+2,n+1]$) by some forward and backward steps that have led to a position $i\in[0,m]$, and let us assume that $\sigma(x_{j+1})=c_i$. Consider the case where $\sigma(x_{j+2})=c_{i+1}$. Then we have $r_{j+1}=r_{i}$. Thus, we made a forward step, and the updated position is $i+1$. Now consider the case where $\sigma(x_{j+2})=c_{i-1}$. Then we have $r_{j+1}=r_{i-1}^-$. Thus, we made a backward step, and the updated position is $i-1$.
Since we have $r_{n+1}=r$ and $\sigma(x_{n+2}) = c_0$, the walk must end at position $0$. The claim follows when we reach $j=n+1$.

\input{appendix/strong_smart_plans.tex}

\subsection{Proof of Lemma~\ref{lem:smart-equivalent_r_atom}}

First, we prove that there cannot be a filter on a constant $b$ different from the input constant $a$. Indeed, let us consider a smart plan $\pi_a$ and its constraint free version $\pi_a'$. The semantics of $\pi_a'$ is $r_1(a, x_1)...r_n(x_{n-1}, x_n)$. We define the database $\mathcal{I}$ as:

\[ \mathcal{I} = r(a, c_0),r_1(a, c_1)...r_n(c_{n-1}, c_n) \]

where $c_1,..., c_n$ are fresh constant different from $a$ and $b$. On $\mathcal{I}$, we have $q(\mathcal{I}) \neq \emptyset$ and $\pi_a'(\mathcal{I}) \neq \emptyset$. So, we must have $\pi_a(\mathcal{I}) = q(\mathcal{I})$. However, if $\pi_a$ contained a filter using a constant $b$, we would have $\pi_a(\mathcal{I}) = \emptyset$ (as $b$ is not in $\mathcal{I}$). This is impossible.

Now, we want to show that the semantics of $\pi_a(x)$ contains at least an atom $r(a, x)$ or $r^-(x, a)$. We still consider a smart plan $\pi_a$ and its constraint free version $\pi_a'$. The semantics of $\pi_a'$ is $r_1(a, x_1)...r_n(x_{n-1}, x_n)$. We define the database $\mathcal{I}$ as:

\begin{align*}
    \mathcal{I} = &r(a, c_0), r_1(a, c_1), r_1(c_1, c_1), ..., r_n(c_1, c_1), r(c_1, c_1),\\
                  &  r_1(c_1, c_1), ..., r_n(c_1, c_2), r(c_1, c_2),\\
                  & r_1(c_2, c_2), ..., r_n(c_2, c_2), r(c_2, c_2), r_1(c_0, c_1), ..., r_n(c_0, c_1), r(c_0, c_1), \\
                  & r_1(c_0, c_2), ..., r_n(c_0, c_2), r(c_0, c_2)
\end{align*}

Let us write the semantics of $\pi_a(x)$ as $r_1(a, y_1)...r_n(y_{n-1}, y_{n})$ where each $y_i$ is either an existential variable, the output variable $x$ or the constant $a$. We call $\mathcal{B}(y_i)$ the set of possible bindings for $y_i$.
We notice that for all $i \in [1, n]$, $y_i$ is a filter to the constant $a$ or $\mathcal{B}(y_i) \neq \{a\}$. This can be seen by considering all transitions for all possible sets of bindings.
We consider the atoms containing the output variable $x$. They are either one or two such atoms. If there is one, it is $r_k(y_{k-1}, x)$ for some $k \in [1, n]$. If $y_{k-1}$ is not a filter to $a$, as we know $\mathcal{B}(y_k) \neq \{a\}$, we have that $\{c_1, c_2\} \subseteq \mathcal{x}$ due to the structure of $\mathcal{I}$. This is a contradiction so $y_{k-1}$ must be a filter to $a$ and $r_k = r$. The same reasoning applies when there are two atoms containing $x$, except that we have one of them which is either $r(a, x)$ or $r^-(x, a)$.

\subsection{Proof of Property~\ref{prop:minimal-filtering-is-path-query}}

Let us decompose $\pi_a$ into a sequence of calls $c_0,..., c_n$. The only call which contains a filter must be the last one (as $x$ only appears in one call). So, we require only one output for the calls $c_0,..., c_{n-1}$ and thus the semantics of there calls is a path query. The last function call is a path function, so the semantics of $\pi_a$ is a path query. We need to show the last atom is either $r(a, x)$ or $r^-(x, a)$. By definition of a minimal filtering plan, the semantics of the plan must contain either $r(a, x)$ or $r(x, a)$ and at most one filter. If the filter is before $x$, then the last atom is $r(a, x)$ and the sub-function can stop there as there are no other filter nor output variable. If the filter is after $x$, we must have $r^-(x, a)$ as the last atom and the sub-function can also stop there.

\subsection{Proof of Theorem~\ref{thm:well-filtering-minimal-construction-smart}}

Let us write $\pi_a'$ the constraint free version of $\pi_a$ which is also the one of $\pi_a^{min}$. We first prove the first point. If $\pi_a^{min}$ is not smart, it means that there exists a database $\mathcal{I}$ such that $q(\mathcal{I}) \neq \emptyset$ and $\pi_a'(\mathcal{I}) \neq \emptyset$ but $\pi_a^{min} \neq q(\mathcal{I})$. As $r(a, x)$ (or $r^-(x, a)$) is in $\pi_a$ and $\pi_a^{min}$ and $\pi_a$ might contain more filters than $\pi_a^{min}$, we have $\pi_a(\mathcal{I}) \subseteq \pi_a^{min}(\mathcal{I}) \subset q(\mathcal{I})$ and thus $\pi_a$ can be smart.

For the second point, we consider that $\pi_a'$ has the form described in Section~\ref{sec:charac-smart-plan}, after the transformation of Property~\ref{prop:minimal-filtering-is-path-query}: $F.B$ where $F$ is a forward path and $B$ is a walk through $F$. Then the correct filters at filter on a variable which is at position $1$ during the backward walk. To see that, we can consider the database $\mathcal{I}$:

\[ \mathcal{I} = r(a, c_0),r_1(a, c_1)...r_n(c_{n-1}, c_n) \]

where the query was $q(x) \leftarrow r(a, x)$ and the semantics of the execution plan was $r_1(a, x_1) \ldots r_n(x_{n-1}, x_n)$. If we follow a proof similar to the one in Section~\ref{sec:proof-weak-smart-complete}, we find that we can keep the result only iff we filter at position $1$ during the backward walk.

\subsection{Proof of Theorem~\ref{Susie-correct}}

Given a Susie plan of the form $\pi=F.F^-.r$, $F^-$ is a walk through $F$. Hence, $F^-.r$ is a walk through $r^-.F$. Hence, the constraint-free version of $\pi$ is a bounded plan (Definition~\ref{def:bounded-plan}). Then, Theorem~\ref{thm:bounded-correct} tells us that $\pi$ is a weakly smart plan. In addition, it ends with $r(a, x)$, so it is smart.

\subsection{Proof of Theorem~\ref{Susie-complete}}

First, let us notice that with these conditions, the plan cannot end by $r^-(x, a)$. Indeed, in this case, it would mean that the two last atoms are $r(y, x).r^-(x, a)$. As no function contains two consecutive atoms $r.r^-$, it means that the last function call is $r^-(x, a)$ and it is useless as we could have created the plan ending by $r(a, x)$ by putting the filter on the previous atom.
Consider any minimal smart plan $P=f_1...f_n$ with forward path $F$ and backward path $B$ for a query $q$. The last function call ($f_n$) will have the form $r_1...r_mq$. Thus, $F$ must start with $r_m^-...r_1^-$ as $f_n$ does not contain loops. Consider the first functions $f_1...f_k$ of $P$ that cover $r_m^-...r_1^-$. $f_k$ must be of the form $r^-_i...r_1^-r_1'...r_l'$ for some $i\in[1,m]$ and some $l\geq 0$. Assume that $l>0$. Since there are no existential variables, we also have a function $f_k'$ of the form $r^-_i...r_1^-$. Then, the plan $f_1...f_{k-1}f_k' f_n$ is smart. Hence, $P$ was not minimal. Therefore, let us assume that $l=0$. Then, $P$ takes the form $P=f_1...f_k f_{k+1}...f_{n-1} f_n$. If $k<n-1$, then we can remove $f_{k+1}...f_{n-1}$ from the plan, and still obtain a smart plan. Hence, $P$ was not minimal. Therefore, $P$ must have the form $P=f_1...f_k f_n$. Since $f_1...f_n$ has the form $r_m^-...r_1^-$, and $f_n$ has the form $r_1...r_mq$, $P$ is a Susie plan.

\subsection{Proof of Theorem~\ref{no-duplicate-ends}}

Assume that there is a weak smart plan $P=f_1...f_n g_1...g_m h_1...h_k$ such that $f_n$ and $g_m$ end at the same position. Then $f_1...f_n h_1...h_k$ is also a weak smart plan. Hence, $P$ is not minimal. Now assume that there is a weak smart plan $P=f_1...f_n g_1...g_m h_1...h_k$ such that $g_1$ and $h_1$ start at the same position. Then again, $f_1...f_n h_1...h_k$ is also a weak smart plan and hence $P$ is not minimal.

\subsection{Proof Of Theorem~\ref{bound-on-plans}}

Theorem~\ref{no-duplicate-ends} tells us that there can be no two positions in the forward path where two function calls end. It means that, at each position, there can be no more than $2\times k$ crossing function calls. 
Therefore, there can be only $M={|\mathcal{F}|}^{2k}$ different states overall. No minimal plan can contain the same state twice (because otherwise, it would be redundant). Hence, there can be at most $M!$ minimal weak smart plans. Indeed, in the worst case, all plans are of length $M$, because if a plan of length $<M$ is weakly smart, then all plans that contain it are redundant.

\subsection{Proof of Theorem~\ref{theorem-algorithm}}

Algorithm~\ref{algo:cikm-algorithm} just calls Algorithm~\ref{search} on all possible starting states. Let us, therefore, concentrate on Algorithm~\ref{search}. This algorithm always maintains one positioned function as the designated forward path function. These designated functions are chained one after the other, and they all point in the forward direction. No other function call can go to a position that is greater than the positions covered by the designated functions (Line~14). Hence, the designated functions form a forward path. All other functions perform a walk in the forward path. Hence, all generated plans are bounded. The plans are also valid execution plans because whenever a function ends in one position, another function starts there (Lines 13, 16, 20). Hence, the algorithm generates only bounded plans, i.e., only weak smart plans.

At any state, the algorithm considers different cases of functions ending and functions starting (Lines~13, 16, 20). There cannot be any other cases, because, in a minimal weak smart plan, there can be at most one single function that starts and at most one single function that ends in any state. In all of the considered cases, all possible continuations of the plan are enumerated (Lines 17, 21, 26). There cannot be continuations of the plan with more than one new function, because minimal weak smart plans cannot have two functions starting at the same state. Hence, the algorithm enumerates all minimal weak smart plans.

We have already seen that the number of possible states is bounded (Theorem~\ref{bound-on-plans}). Since the algorithm keeps track of all states that it encounters (Line~2), and since it never considers the same state twice (Line~1), it has to terminate. Furthermore, since $M$ bounds the size of $H$, the algorithm runs in time $\mathcal{O}(M!)$.

\subsection{Proof of Theorem~\ref{find-single}}

If we are interested in finding only a single plan, we replace $H$ by a set in Line~3 of Algorithm~\ref{algo:cikm-algorithm}. We also remove the pop operation in Line~29 of Algorithm~\ref{search}. In this way, the algorithm will never explore a state twice. Since no minimal weak smart plan contains the same state twice, the algorithm will still find a minimal weak smart plan, if it exists (it will just not find two plans that share a state). Since there are only $|\mathcal{F}|^{2k}$ states overall (Theorem~\ref{bound-on-plans}), the algorithm requires no more than $O(|\mathcal{F}|^{2k})$ steps.

\input{appendix/discussions.tex}

\input{appendix/non_context_freeness.tex}

%% file: appendix/susie.tex
\section{Susie}\label{Susie}

Our first approach to generate bounded plans is inspired by Susie~\cite{susie}. We call the resulting plans \emph{Susie plans}.

\begin{Definition}[Susie Plan]
A Susie plan for a query $q(x) \leftarrow r(a, x)$ is a plan whose semantics is of the form $F.F^-.r$, where $F$ is a path query, $F^-$ is the reverse of $F$ (with all relations inverted), $F^-.r$ is generated by a single function call and the last atom is $r(a, x)$.
\end{Definition}

\noindent Example~\ref{ex1} (shown in Figure~\ref{fig:ex1}) is a simple Susie plan. We have:

\begin{Theorem}[Correctness]\label{Susie-correct}
    Every Susie plan is a smart plan.
\end{Theorem}

\begin{Theorem}[Limited Completeness]\label{Susie-complete}
Given a set of path functions so that (1) no function contains existential variables and (2) no function contains a loop (i.e., two consecutive relations of the form $r.r^-$), the Susie plans include all minimal filtering smart plans.
\end{Theorem}

\noindent Minimal filtering plans are introduced in \cite{romero2020equivalent} and we recall the definition in Appendix~\ref{sec:characterising-smart-plans}. Intuitively, Theorem~\ref{Susie-complete} means that Susie plans are sufficient for all cases where the functions do not contain existential variables. If there are existential variables, then there can be minimal smart plans that Susie will not find.
An example is shown in Figure~\ref{fig:cikm-plan-ex1}, if we assume that the middle variable of \emph{get\-Professional\-Address} existential.

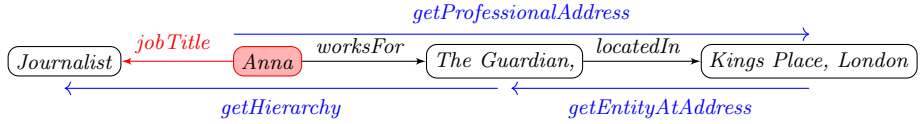
\begin{figure*}[t]
    \centering
    \input{images/CIKM2018/fig-plan-album2}
    \caption{A smart plan for the query \textit{job\-Title}(\textit{Anna}, $?x$), which Susie will not find.\label{fig:cikm-plan-ex1}}
\end{figure*}

The Susie plans for a query $q$ on a set of functions $\mathcal{F}$ can be generated very easily: It suffices to consider each function $f=r_1....r_n$ that ends in $r_n=q$, and to consider all path function $f_1...f_m$ that have the semantics $r^-_{n-1}...r^-_1$. This algorithm runs in $O(|\mathcal{F}|^{k+1})$, where $k$ is the maximal length of a function, and is typically small.

%% file: appendix/bound_weak_smart_plans.tex
\section{Bounding the Weakly Smart Plans}\label{sec-bound}

We would like to generate not just the Susie plans, but all minimal smart plans. This is possible only if their number is finite. This is indeed the case, as we shall see next: Consider a weakly smart plan for a query $q$. Theorem~\ref{thm:bounded-complete} tells us that the consequences of the plan are of the form $r_0...r_k...r_n q$, where $r_0...r_k$ is the forward path, and $r_{k+1}...r_n$ is a walk in the forward path. Take some function call $f=r_i...r_j$ of the plan. If $j\leq k$, we say that $f$ \emph{starts} at position $i$, that it \emph{ends} at position $j$, and that it \emph{crosses} positions $i$ to $j$. For example, in Figure~\ref{fig:cikm-plan-abstract}, the forward path is $F=rst$. The function $f_1$ starts at position 0, ends at position 3, and crosses the positions 0 to 3. If $i\geq k$, then $r_i$ was generated by a step in the backward path. We say that $f$ starts at the position before that step, and ends at the updated position after the step that produced $r_j$. In Figure~\ref{fig:cikm-plan-abstract}, $f_2$ starts at position 3, ends at position 1, and crosses 3 to 1; $f_3$ starts at position 1, ends at position 2, and crosses 1 to 2; $f_4$ starts at 2, ends at -1, and crosses 2 to -1. Our main insight is the following:

\begin{Theorem}[No Duplicate Ends]\label{no-duplicate-ends}
    Given a set of relations $\mathcal{R}$, a query $q(x) \leftarrow r(a, x), r \in \mathcal{R}$, and a set of path function definitions $\mathcal{F}$, let $\pi$ be a minimal weakly smart plan for $q$. There can be no two function calls in $\pi$ that end at the same position, and there can be no two function calls that start at the same position.
\end{Theorem}

\noindent This theorem is easy to understand: If two functions were ending or starting at the same position, the plan would not be minimal as we could remove the function calls between them. To turn this theorem into a bound on the number of possible plans, we need some terminology first. Let us consider the positions on the forward path one by one. Each position is crossed by several functions. Let us call the triple of a function $f=r_1...r_n$, a position $i \in[1,n]$, and a direction (forward or backward) a \emph{positioned function}. For example, in Figure~\ref{fig:cikm-plan-abstract}, we can say that position 2 in the forward path is crossed by the positioned function $\langle f_1,3,\textit{forward}\rangle$. Let us call the set of positioned functions at a given position in the forward path a \emph{state}. For example, in Figure~\ref{fig:cikm-plan-abstract} at position 2, we have the state $\{\langle f_1,3,\textit{forward}\rangle, \langle f_2,2,\textit{backward}\rangle, \langle f_3,2,\textit{forward}\rangle, \langle f_4,1,$ $\textit{backward}\rangle\}$. 

We first observe that a state cannot contain the same positioned function more than once. If it contained the same positioned function twice, then the plan would not be minimal. Furthermore, Theorem~\ref{no-duplicate-ends} tell us that there can be no two functions that both end or both start in a state. Finally, a plan cannot contain the same state twice, because otherwise, it would not be minimal. This leads to the following bound on the number of plans:

\begin{Theorem}[Bound on Plans]\label{bound-on-plans}
    Given a set of relations $\mathcal{R}$, a query $q(x) \leftarrow r(a, x), r \in \mathcal{R}$, and a set of path function definitions $\mathcal{F}$, there can be no more than $M!$ minimal weakly smart plans, where $M={|\mathcal{F}|}^{2k}$ and $k$ is the maximal number of atoms in a function.
\end{Theorem}

\noindent This bound is very pessimistic: In practice, the succession of states in very constrained and thus, the complete exploration is quite fast, as we showed in Section~\ref{sec:cikm-experiments}.

%% file: appendix/only_minimal.tex
\section{We can turn non-minimal plans into minimal ones}

The following property confirms the intuition that if we have a non-minimal execution plan, then we can turn it into a minimal one:

\begin{Property}
Let $q$ be an atomic query and $\pi_a$ a non-redundant weakly smart plan. Then, either $\pi_a$ is minimal, or one can extract a minimal weakly smart plan from $\pi_a$.
\end{Property}

\noindent This means that even if we can generate only the minimal weakly smart plans (and not all weakly smart plans), we will be able to generate a plan if a weakly smart plan exists at all.

%% file: appendix/smart_generation_adaptation.tex
\section{Generating Smart Plans}
\label{sec:generate-smart-plans}

To generate smart plans instead of weakly smart plans, we will adapt Algorithm~\ref{algo:cikm-algorithm}. The intuition is simple: if the query is $q(x) \leftarrow r(a,x)$, and if we found a weakly smart plan that contains the atom $r(y,x)$, then we have to add a filter $y=a$ to make the plan smart. For this, we have to make sure that $y$ is an output variable.
In Appendix~\ref{sec:characterising-smart-plans}, we show why this is sufficient to make a weakly smart plan smart.
Besides, 
we show that we can restrict ourselves to smart plans whose semantics are path queries ending in $r(y, x)$ or $r^-(x, y)$. We now give more details about how we adapt the generation of weakly smart plans to generate also smart plans.

We begin by considering the case where the semantics of the plan is a path query ending in $r(y, x), y=a$. 
We have to modify Algorithm~\ref{algo:cikm-algorithm} in such a way that the last two variables of the sub-function that starts the search must be output variables. If the algorithm gives no result, we can conclude that no smart plan exists (because there exists no weakly smart plan). If the algorithm gives a result, then we know that we can add the filter $y=a$. 

Now let us consider the case where the semantics of the plan is a path query ending with $r^-(x, y), y=a$. In this case, the restriction to sub-functions will remove the atom $r^-(x, a)$. Therefore, we have to adapt the initialisation of the algorithm in Section~\ref{sec-algorithm}. We have two cases to consider:
\begin{itemize}
    \item The last function call contains only one atom: $r^-(x, a)$. This happens when there is a sub-function with one atom $r^-$ with one input and one output variable. In this case, we just have to find a weakly smart plan as in Section~\ref{sec-algorithm} and add this function at the end.
    \item The last call contains more than one atom. In this case, we have to look at all sub-functions ending in $r(z, y).r^-(y, x)$ where $y$ and $x$ are output variables. We then continue the initialisation as if the function was ending in $r(z, y)$, ignoring the last atom $r^-(y, x)$. 
\end{itemize}
\noindent If this generation algorithm gives no result, we know there is no smart plan. Otherwise, we apply the filter on the last atom of the semantics to get $r^-(x, y), y=a$. This is possible because we made sure that $y$ is an output variable.

We give the complete new algorithm in Algorithm~\ref{algo:cikm-algorithm-smart}.

\begin{algorithm}\caption{FindSmartPlans}\label{algo:cikm-algorithm-smart}
    \DontPrintSemicolon
    \KwData{Query $q(a) \leftarrow r(a, x)$, set of path function definitions and all their sub-functions $\mathcal{F}$}
    \KwResult{Prints minimal smart plans}
    \If{$\exists f(x, y)=r(x, y) \in \mathcal{F}$}{
        print($f$)
    }
    $H\gets Stack()$\;
    \ForEach{$f(x, y, z)=r_1...r_n(x, y).r(y, z) \in \mathcal{F}$}{
        \ForEach{$f' \in \mathcal{F}$ consistent with $r_n^-...r_1^-$}{
            \ForEach{Weak Smart Plan in search($\{\langle f,n,backward\rangle, \langle f',1,forward\rangle^*\}$, $H$, $\mathcal{F}$)}{
                Add a filter to create $r(a, x)$
            }
        }
    }
    \ForEach{$f(x, y, z)=r_1...r_n.r(x, y).r^-(y, z) \in \mathcal{F}$}{
        \ForEach{$f' \in \mathcal{F}$ consistent with $r_n^-...r_1^-$}{
            \ForEach{Weak Smart Plan in search($\{\langle f,n,backward\rangle, \langle f',1,forward\rangle^*\}$, $H$, $\mathcal{F}$)}{
                Add a filter to create $r^-(x, a)$
            }
        }    
    }
    \ForEach{$f(x, y)=r^-(x, y) \in \mathcal{F}$}{
        FindMinimalWeakSmartPlans($q$, $\mathcal{F}$) + $f(x, a)$
    }
\end{algorithm}

%% file: appendix/strong_smart_plans.tex
\subsection{Characterising Smart Plans}
\label{sec:characterising-smart-plans}

When we want to show that a plan is smart, we must first show that its constraint-free version is weakly smart. Thus, the results of Section~\ref{sec:charac-smart-plan} can be applied to determine the shape of a smart plan. What remains to find are the constrains (or filters) that we must apply to the execution plan in order to make it smart.

As in Romero et al.~\cite{romero2020equivalent}, we are going to exploit the notion of well-filtering plan (Definition~\ref{def:wellfiltering} and minimal well-filtering plan (Definition~\ref{def:minimal-filtering-plan}).

\begin{Definition}[Well-Filtering Plan]
\label{def:wellfiltering}
Let $q(x) \leftarrow r(a, x)$ be an atomic query. An execution plan $\pi_a(x)$ is said to be \emph{well-filtering for $q(x)$} if all filters of the plan are on the constant $a$ used as input to the first call and the semantics of $\pi_a$ contains at least an atom $r(a, x)$ or $r^-(x, a)$, where $x$ is the output variable.
\end{Definition}

\begin{Definition}[Minimal Filtering Plan]
\label{def:minimal-filtering-plan}
Given a well-filtering plan $\pi_a(x)$ for an atomic query $q(a,x) \leftarrow r(a,x)$, let the \textit{minimal filtering plan associated to $\pi_a(x)$} be the plan $\pi'_a(x)$ that results from removing all filters from $\pi_a(x)$ and doing the following:
\begin{itemize}
    \item We take the greatest possible call $c_i$ of the plan, and the greatest possible output variable $x_j$ of call~$c_i$, such that adding a filter on~$a$ to variable $x_j$ of call~$c_i$ yields a well-filtering plan, and define $\pi_a'(x)$ in this way.
    \item If this fails, i.e., there is no possible choice of $c_i$ and $x_j$, then we leave $\pi_a(x)$ as-is, i.e., $\pi'_a(x) = \pi_a(x)$.
    \end{itemize}
\end{Definition}

We have a similar result than Romero et al.~\cite{romero2020equivalent}:

\begin{Lemma}\label{lem:smart-equivalent_r_atom}
   Given an atomic query $q(a,x) \leftarrow r(a,x)$ and a set of path functions $\mathcal{F}$, any smart plan constructed from sub-functions of $\mathcal{F}$ must be well-filtering.
\end{Lemma}

We have the equivalent of Proposition~\ref{prop:unconstraint-is-path-query} for minimal filtering plans:

\begin{Property}
    \label{prop:minimal-filtering-is-path-query}
    Let $q(x) \leftarrow r(a, x)$ be an atomic query. Let $\pi_a$ be a minimal filtering execution plan associated to $q$. Then the semantics of the sub-function transformation $\mathcal{P}(\pi_a)$ is a path query where the last atom is either $r(a, x)$ or $r^-(x, a)$.
\end{Property}

We also have the equivalent of a theorem in Romero et al.~\cite{romero2020equivalent}.

\begin{restatable}{Theorem}{thmwellfilteringminimalconstructionsmart}\label{thm:well-filtering-minimal-construction-smart}
Given a query $q(x) \leftarrow r(a,x)$, a well-filtering plan $\pi_a$ and the associated minimal filtering plan $\pi_a^{min}$:
\begin{itemize}
    \item If $\pi_a^{min}$ is not smart, then neither is $\pi_a$. 
    \item If $\pi_a^{min}$ is smart, then we can determine in polynomial time if $\pi_a$ is also smart.
\end{itemize}
\end{restatable}

This theorem tells us that it is enough to find minimal filtering smart plans.

Given a plan $\pi_a$, let us consider that we construct the plan $\pi_a'$ as described in Property~\ref{prop:minimal-filtering-is-path-query} by using only sub-function calls associated with the output variables. If $\pi_a'$ is smart, it means that its constraint-free version is weakly smart and thus we follow the description in Section~\ref{sec:charac-smart-plan}: a forward path $F$ followed by a walk in $r^-F$ ending at position 0. Then, from this construction, one can deduce the minimal-filtering smart plans. First, we create plans ending by $r(a, x)$ by adding a filter on the last atom of $pi_a'$. Second, we create plans ending by $r^-(x, a)$ by adding a new atom $r^-(x, a)$ to the semantics of the plan.

Thus, if we have an algorithm to find weakly smart plans, it can easily be extended to generate smart plans as we will see in Section~\ref{sec:generate-smart-plans}.

%% file: appendix/discussions.tex
\section{Discussion}

In Romero et al~\cite{romero2020equivalent}, the problem of finding equivalent rewritings is reduced to the problem of finding a word in a context-free grammar. This type of grammar makes it possible to check the emptiness of the language in polynomial time or to compute the intersection with a regular language.
The emptiness operation can be used to check in advance if an equivalent rewriting exists, and the intersection can be used to restrict the space of solutions to valid execution plans. 

In our problem, we also define a language: the language of bounded plans. Unfortunately, it turns out that this language is not context-free. We give a proof in Appendix~\ref{non-context-free}.
This has two consequences: First, it is not trivial to find the intersection with the valid execution plans as it was done in~\cite{romero2020equivalent}.
Second, it explains the exponential complexity bound for our algorithm: our language is more complicated than a context-free grammar, and therefore, there is a priori no reason to believe that the emptiness problem is polynomial.



%% file: appendix/non_context_freeness.tex
\subsection{Non Context-Freeness}\label{non-context-free}

Here we prove that we cannot represent smart plans with a context-free grammar in the general case. To do so, we will use a generalization of Olgen's Lemma presented in \cite{Bader82}.

\begin{Lemma}[Bader-Moura's Lemma]
For any context-free language $L$, $\exists n \in \mathcal{N}$ such that $\forall z \in L$, if $d$ positions in $z$ are ``distinguished'' and $e$ positions are "excluded", with $d > n^{e+1}$, then $\exists u, v, w, x, y$ such that $z = uvwxy$ and:
\begin{enumerate}
\item $vx$ contains at least one distinguished position and no excluded positions
\item if $r$  is the number of distinguished positions and $s$ the number of excluded positions in $vwx$, then $r \leq  n^{s+1}$
\item $\forall i \in \mathcal{N}, u.v^i.w.x^i.y \in L$
\end{enumerate}
\end{Lemma}

\begin{Theorem}
For $|\mathcal{R}| > 4$, the smart plans are not context-free.
\end{Theorem}

\begin{proof}

We begin giving a property over words of our grammar.

\begin{Lemma}
For each relation $r$ different from the query, the number of $r$ in a word is equal to the number of $r^-$ in this word.
\end{Lemma}

Let us suppose our grammar is context-free. We define $n$ according to the Barder-Moura's lemma. Let $a, b, c$ be three distinct relations in $|R|$ (they exist as $|\mathcal{R}| > 4$). Let $k = n^{10 + 1}$. We consider the word $z = a.b^k.a.c.c^-.a^-.b^{-k}.a^-.a.b^k.a.a^-.b^{-k}.a^-$ in which we distinguish all the $b$ and $b^-$ and exclude all the $a$, $a^-$, $c$, and $c^-$.

We write $z = u.v.w.x.y$. As $v.x$ contains no excluded positions, $v$ contains only $b$s or only $b^-$s (and the same is right for $x$ as well). As we must keep the same number of $b$ and $b^-$ according to the previous lemma, either $v = b^j$ and $x = b^{-j}$ or $v = b^{-j}$ and $x = b^{j}$ (for some $j \in \mathcal{N}$).

In $z$, it is clear that $a.b^k.a.c$ is the forward path $F$ as $c$ appears only once. $c^-$ and the last $a^-.b^{-k}.a^-$ are generated by the $B$. $a^-.b^{-k}.a^-.a.b^k.a$ is generated by $L$.

The forward-path defines the number of consecutive $b$s between two $a$s (and so the number of consecutive $b^-$s between two $a^-$s). In these conditions, it is impossible to make the four groups of $b$ and $b^-$ vary the same way, so it is impossible to define $v$ and $x$, and so the grammar is not context-free.

\end{proof}